\documentclass[pdflatex,sn-mathphys-num]{sn-jnl}

\usepackage{graphicx}%
\usepackage{multirow}%
\usepackage{amsmath,amssymb,amsfonts}%
\usepackage{amsthm}%
\usepackage{mathrsfs}%
\usepackage[title]{appendix}%
\usepackage{xcolor}%
\usepackage{textcomp}%
\usepackage{manyfoot}%
\usepackage{booktabs}%
\usepackage{algorithm}%
\usepackage{algorithmicx}%
\usepackage{algpseudocode}%
\usepackage{listings}%
\usepackage{makecell, pdflscape}
\usepackage{tabularx}
\usepackage{comment}
\usepackage[dvipsnames]{xcolor}
\usepackage{trimclip,lipsum}
\usepackage{epsfig,mathtools}
\usepackage{caption, subcaption}
\usepackage[table]{xcolor}

\theoremstyle{thmstyleone}%
\newtheorem{theorem}{Theorem}
\theoremstyle{thmstyletwo}%

\theoremstyle{thmstylethree}%

\begin{document}

\title[]{The Role of Tick Questing Behavior and Host-Finding Success in \textit{Borrelia burgdorferi} Transmission in the Eastern United States}


\author*[1]{\fnm{Carli} \sur{Peterson}}\email{cxp4986@mavs.uta.edu}

\author[1]{\fnm{Erika} \sur{Gallo}}\email{erika.gallo@uta.edu}

\author[1,2]{\fnm{Christopher} \sur{Kribs}}\email{kribs@uta.edu}

\affil*[1]{\orgdiv{Department of Mathematics}, \orgname{The University of Texas at Arlington}, \orgaddress{\street{411 S Nedderman Dr}, \city{Arlington}, \postcode{76019}, \state{TX}, \country{USA}}}

\affil[2]{\orgdiv{Department of Teacher and Administrator Preparation}, \orgname{The University of Texas at Arlington}, \orgaddress{\street{701 Planetarium Pl}, \city{Arlington}, \postcode{76019}, \state{TX}, \country{USA}}}


\abstract{The marked difference in reported human Lyme disease cases between the northeastern and southeastern United States remains a significant concern. Differences in biodiversity and the questing behavior of \textit{Ixodes scapularis}, the primary vector of \textit{Borrelia burgdorferi}, the causative agent of Lyme disease, have been proposed as key drivers of geographic variation in disease incidence. In this study, we investigate how these differences, together with ratio-dependent tick-host contact rates, contribute to the north-south gap in Lyme disease cases. Building on our previous work, we develop a discrete-time system of difference equations to model infection dynamics between ticks and their sylvatic hosts. We derive an expression for the basic reproductive number and vary key parameters to assess their impact on infection persistence. We further examine the effects of these mechanisms on nymphal infection prevalence (NIP) and the density of infected nymphs (DIN), two key metrics for predicting disease risk. Our results indicate that under baseline conditions, southern host communities exhibit lower infection prevalence and reduced potential for pathogen invasion, consistent with observed regional patterns. However, southern host communities can support sustained transmission with different questing behaviors. Notably, this qualitative north-south difference emerges using parameter values obtained from the literature, without calibrating the model to reproduce geographic differences in disease prevalence. In addition, ratio-dependent host-finding success shows a saturating effect on the basic reproductive number, compared to the traditional assumption of density-independent host-finding success. These results highlight the importance of biodiversity, tick questing behavior, and using multiple metrics when assessing Lyme disease risk.} 


\maketitle

\section{Introduction}\label{introduction}

\textit{Borrelia burgdorferi}, the bacterium responsible for Lyme disease, is maintained in nature through a complex multi-host, multi-stage transmission cycle involving larval, nymphal, and adult ticks. In the eastern United States, where Lyme disease is most prevalent, \textit{Ixodes scapularis} is the primary vector of \textit{B. burgdorferi} \citep{cdc_lyme}. Ticks typically acquire the pathogen during their larval or nymphal stages by feeding on infected reservoir hosts; once infected, they can then transmit the bacterium during subsequent feedings. Among the tick life stages, nymphs are considered the greatest risk to humans because of their small size and the timing of their peak activity \citep{logiudice2003ecology, lou2014impact}. As a result, nymphal infection prevalence (NIP) and the density of infected nymphs (DIN) are key ecological indicators used to estimate human risk of Lyme disease \citep{logiudice2003ecology, lou2014impact}. Interestingly, these measures vary across geographic regions in the eastern United States \cite{peterson2026host}.

Most reported Lyme disease cases in the eastern U.S. occur in the northeastern and upper midwestern regions \citep{cdc_lyme, cdc_lyme_surveillance}, despite the presence of the primary vector and many of the same sylvatic host species found throughout the region \citep{clark2002prevalence, levin1996reservoir}. This discrepancy has motivated much ecological and epidemiological research aimed at explaining this geographic variation in disease incidence. One common hypothesis emphasizes the role of biodiversity, suggesting that the richer and more even host communities in the southeastern U.S. reduce transmission by increasing the proportion of incompetent hosts available to feeding ticks (i.e., the ``dilution effect'') \citep{ostfeld2000biodiversity}. Supporting this hypothesis, \textit{I. scapularis} nymphs in the southeastern U.S. have very low infection prevalence, reported to be near 0\% in some locations, compared to approximately 24\% in northeastern populations \citep{ginsberg2021lyme, foster2023prevalence}. Although biodiversity likely contributes to this pattern \citep{logiudice2003ecology}, the drastic difference in NIP raises the question of whether biodiversity alone can fully explain the low prevalence of \textit{B. burgdorferi} in the southeastern U.S.

Another hypothesis involves regional variation in tick questing behavior. Southern \textit{I. scapularis} ticks tend to remain below the leaf litter for longer periods, likely due to selective pressure to avoid desiccating conditions \citep{ginsberg2017environmental}, though some studies suggest the difference is due to genetic variation between northern and southern populations \citep{arsnoe2019nymphal}. In either case, ticks that quest below the leaf litter are more likely to encounter reptiles, which are generally poor reservoirs for \textit{B. burgdorferi} \citep{arsnoe2019nymphal, giery2007role}. In contrast, northern ticks typically quest higher above the leaf litter and are more likely to encounter more competent hosts, such as small mammals \citep{arsnoe2019nymphal, logiudice2003ecology}. These differences in host encounter patterns could lead to variation in the probability that ticks become infected from north to south. 

Mathematical models provide a useful tool for investigating these hypotheses without the need for extensive data collection and have played a central role in understanding Lyme disease ecology. Transmission models for \textit{B. burgdorferi} have been well studied and range from deterministic differential equation systems to discrete-time, stochastic, and simulation-based frameworks \citep{ogden2005dynamic, gaff2020lymesim, caraco2002stage, zhang2023coupled}. Many published studies use ordinary differential equations and incorporate seasonal variation through non-autonomous systems or delay differential equations \citep{nguyen2019modeling, husar2024lyme}. While these approaches have provided important insight, the resulting models often become highly complex, which can limit the use of traditional analytical tools and the ability to draw biological conclusions. Discrete-time models provide an alternative framework that is well suited to the tick life cycle, since feeding, molting, and reproduction occur in distinct seasonal stages. Several discrete-time models have been proposed to describe tick infection dynamics \citep{porco1999mathematical, van1995controlling, schauber2002modeling, ackleh2022modeling}; however, many do not fully exploit the ability to sequence biological events before collapsing them into a single yearly time step to generate the resulting difference equation system (see, for example, \cite{carrera2020cost}).

Regardless of the modeling framework, many existing studies assume either frequency-dependent or density-dependent contact rates between ticks and hosts \citep{lou2017modeling}. While these assumptions can be reasonable, they may not fully capture the observed saturation that occurs when host availability is limited or tick densities are high \citep{ogden2005dynamic, mount1997simulation, ginsberg2020local}. In addition, questing behavior is rarely represented explicitly, despite evidence that differences in questing and host encounter patterns may contribute to regional variation in Lyme disease risk.

To begin addressing these limitations, we previously developed a discrete-time model incorporating host-to-tick ratio-dependent host-finding success and explicit questing behavior \citep{peterson2026model}. In that work, we examined baseline tick and host population dynamics in the absence of pathogen transmission. In the current study, we extend this framework to include the transmission dynamics of \textit{B. burgdorferi} between \textit{I. scapularis} and their sylvatic hosts, with the goal of assessing the effect of ratio-dependent host-finding success and tick questing behavior on pathogen transmission. Specifically, we aim to answer: \textit{How does ratio-dependent host-finding success influence transmission outcomes?} and \textit{To what extent can regional differences in questing behavior and biodiversity explain the geographic distribution of Lyme disease cases in the eastern United States?} By comparing the effects of ratio-dependent host-finding success, questing behavior, and host community composition, our model provides insight into which ecological mechanisms shape pathogen transmission and contribute to human disease risk. 

\section{Model Development}
\label{development}
In this section, we build on the demographic model developed in \cite{peterson2026model}, which describes the population dynamics of \textit{I. scapularis} and six representative host groups in both the northeastern and southeastern U.S. Unless otherwise noted, the host classifications and demographic assumptions remain unchanged. We therefore detail only those components needed to introduce pathogen transmission. Briefly, the host groups for the northeastern region are
\begin{itemize}
    \item Group 1: White-footed mice,
    \item Group 2: Chipmunks,
    \item Group 3: Short-tailed shrews, \textit{Sorex} shrews, and red-backed voles
    \item Group 4: Squirrels,
    \item Group 5: Five-lined skinks and eastern fence lizards
    \item Group 6: Deer, opossums, and raccoons,
\end{itemize}
and for the southeastern set:
\begin{itemize}
    \item Group 1: White-footed mice and cotton mice,
    \item Group 2: Chipmunks,
    \item Group 3: Short-tailed shrews and \textit{Sorex} shrews,
    \item Group 4: Squirrels,
    \item Group 5: Five-lined skinks, eastern fence lizards, southeastern five-lined skinks, and ground skinks,
    \item Group 6: Deer, opossums, and raccoons.
\end{itemize}

\subsection{Modeling Assumptions}\label{assumptions}
To model the transmission dynamics of \textit{B. burgdorferi}, we use the following set of assumptions (note that the demographic assumptions are listed in \cite{peterson2026model}):
\begin{itemize}
    \item Hosts that become infected during a given questing season are not assumed to be infectious to ticks feeding later within that same season.
    \item Hosts in group 6 are considered incompetent due to their low estimated reservoir competence \citep{logiudice2003ecology}.
    \item Infection does not affect host or tick death rates \citep{voordouw2015lyme}.
    \item Ticks and hosts do not recover from infection \citep{barbour1993biological, ostfeld2000biodiversity}.
    \item Infection is transmitted transstadially in ticks with no transovarial or co-feeding transmission \citep{rollend2013transovarial, kurokawa2020interactions, magnarelli1987transovarial,voordouw2015co, piesman2001efficacy}.
\end{itemize}

Because previous modeling studies have included other transmission pathways in their work, the last point mentioned here warrants further discussion and justification. We assume that transmission of \textit{B. burgdorferi} occurs only through horizontal transmission during tick blood meals. Although co-feeding and transovarial transmission have also been proposed as viable infection routes, neither appears to be as efficient as horizontal transmission. For example, co-feeding transmission seems to play a larger role in European systems involving \textit{B. afzelii} and \textit{I. ricinus} than in North American systems with \textit{B. burgdorferi} and \textit{I. scapularis} \citep{voordouw2015co}. Experimental studies further suggest that co-feeding transmission rates are relatively low (approximately one-twentieth of horizontal transmission rates) and typically occur only under unrealistic conditions such as extremely high tick burdens \citep{piesman2001efficacy}. In addition, the long-term viability of spirochetes acquired during co-feeding remains unclear \citep{voordouw2015co}. Similarly, transovarial transmission (from infected female to eggs) is primarily associated with \textit{B. miyamotoi} rather than \textit{B. burgdorferi} \citep{rollend2013transovarial} and is rare in nature \citep{magnarelli1987transovarial}. Consequently, in our system, transmission dynamics of \textit{B. burgdorferi} are driven by host infectivity and tick feeding interactions. 

\subsection{The Mathematical Model}\label{model}

Incorporating the assumptions outlined in Section \ref{assumptions}, we develop the following system of difference equations to model the transmission dynamics of \textit{Borrelia burgdorferi}:

\begin{equation}\label{fullsystem}
\begin{aligned}
    L_S\left(t+1\right)&=B(t)s_LS_L(t)\left(1-\sum_{i=1}^5 p_{i,L}(t) \beta_{i,L} \frac{H_{i_I}\left(t\right)}{H_i\left(t\right)+\frac{\Lambda_i}{k_i}}\right)\\
    L_I\left(t+1\right)&=B(t)s_LS_L(t)\left(\sum_{i=1}^5 p_{i,L}(t) \beta_{i,L} \frac{H_{i_I}\left(t\right)}{H_i\left(t\right)+\frac{\Lambda_i}{k_i}}\right)\\
    N_S\left(t+1\right)&=s_NS_N(t)L_S\left(t\right)\left(1-\sum_{i=1}^5p_{i,N}(t)\beta_{i,N}\frac{H_{i_I}\left(t\right)e^{-\mu_i}}{H_i\left(t\right)}\right)\\
    N_I\left(t+1\right)&=s_NS_N(t)\left(L_I\left(t\right)+L_S\left(t\right)\left(\sum_{i=1}^5p_{i,N}(t)\beta_{i,N}\frac{H_{i_I}\left(t\right)e^{-\mu_i}}{H_i\left(t\right)}\right)\right)\\
    A_S\left(t+1\right)&=s_AS_A(t)s_NS_N(t)L_S\left(t\right)\left(1-\sum_{i=1}^5p_{i,N}(t)\beta_{i,N}\frac{H_{i_I}\left(t\right)e^{-\mu_i}}{H_i\left(t\right)}\right)\\
    A_I\left(t+1\right)&=s_AS_A(t)s_NS_N(t)\left(L_I\left(t\right)+L_S\left(t\right)\left(\sum_{i=1}^5p_{i,N}(t)\beta_{i,N}\frac{H_{i_I}\left(t\right)e^{-\mu_i}}{H_i\left(t\right)}\right)\right)\\~\\
    H_{i_S}\left(t+1\right)&=\left(\left(H_{i_S}(t)+\frac{\Lambda_i}{2}\left(1+e^{\frac{\mu_i}{2}}\right)\right)e^{-\mu_i}\right)e^{-\frac{p_{i,N}(t)\alpha_{i}s_NS_N(t)L_I\left(t\right)}{H_i\left(t\right)}}\quad \text{for $i=1,2,3,4$}\\
    H_{i_I}\left(t+1\right)&=\left(H_{i_I}\left(t\right)+\left(H_{i_S}(t)+\frac{\Lambda_i}{2}\left(1+e^{\frac{\mu_i}{2}}\right)\right)\left(1-e^{-\frac{p_{i,N}(t)\alpha_{i}s_NS_N(t)L_I\left(t\right)}{H_i\left(t\right)}}\right)\right)e^{-\mu_i} \quad \text{for $i=1,2,3,4$}\\
    H_{5_S}\left(t+1\right)&=\left(\left(H_{5_S}\left(t\right)+\Lambda_5\right)e^{-\mu_5}\right)e^{-\frac{p_{5,N}(t)\alpha_{5}s_NS_N(t)L_I\left(t\right)}{H_5\left(t\right)}}\\
    H_{5_I}\left(t+1\right)&=\left(H_{5_I}\left(t\right)+\left(H_{5_S}\left(t\right)+\Lambda_5\right)\left(1-e^{-\frac{p_{i,N}\left(t\right)\alpha_is_NS_N\left(t\right)L_I\left(t\right)}{H_5\left(t\right)}}\right)\right)e^{-\mu_5}\\
    H_6\left(t+1\right)&=\left(H_6(t)e^{-\frac{\mu_6}{2}}+\Lambda_6\right)e^{-\frac{\mu_6}{2}}
\end{aligned}
\end{equation}
where 
\begin{align*}
B(t)&=\frac{r\frac{A\left(t\right)}{2}}{\frac{\frac{A\left(t\right)}{2}}{H_6\left(t\right)}+K}\\
S_L(t) &= 1 - e^{- \frac{z_L}{B(t)}\left(\sum_{i=1}^5 b_{i,L} \left(H_i(t) + \frac{\Lambda_i}{k_i}\right) + b_{6,L} \left(H_6(t)e^{- \mu_6/2} + \Lambda_6\right)\right)}, \\
S_N(t) &= 1 - e^{- \frac{z_N}{L(t)}\left(\sum_{i=1}^5 b_{i,N}\sigma_{i,N} H_i(t) + b_{6,N}\sigma_{6,N} H_6(t)e^{- \mu_6/2}\right)}, \\
S_A(t) &= 1 - e^{- \frac{z_A}{s_NS_N(t)L(t)}b_{6,A}\sigma_{6,A} H_6(t)}, \\
p_{i,L}(t) &= \frac{b_{i,L} \left(H_i(t) + \frac{\Lambda_i}{k_i}\right) e^{- \mu_i/2}}{\sum_{j=1}^5 b_{j,L} \left(H_j(t) + \frac{\Lambda_j}{k_j}\right) e^{- \mu_j/2} + b_{6,L} H_6(t)}, \quad i = 1,\dots,5, \\
p_{6,L}(t) &= \frac{b_{6,L} H_6(t)}{\sum_{j=1}^5 b_{j,L} \left(H_j(t) + \frac{\Lambda_j}{k_j}\right) e^{- \mu_j/2} + b_{6,L} H_6(t)}, \\
p_{i,N}(t) &= \frac{b_{i,N}\sigma_{i,N} H_i(t)}{\sum_{j=1}^5 b_{j,N}\sigma_{j,N} H_j(t) + b_{6,N}\sigma_{6,N} H_6(t)e^{- \mu_6/2}}, \quad i = 1,\dots,5, \\
p_{6,N}(t) &= \frac{b_{6,N}\sigma_{6,N} H_6(t)e^{- \mu_6/2}}{\sum_{j=1}^5 b_{j,N}\sigma_{j,N} H_j(t) + b_{6,N}\sigma_{6,N} H_6(t)e^{- \mu_6/2}}, \\
k_i &= 2 \quad \text{for } i = 1,2,3,4; \quad k_5 = 1, \\
b_{i,j} &= \hat{b}_{i,j}c \quad \text{for } i = 1,\dots,4,6; \quad b_{5,j} = \hat{b}_{5,j}(1 - c) \quad \text{and} \quad  j\in\{L,N\}.
\end{align*}

Table~\ref{param_table} details the model parameters and state variables. The model tracks larvae ($L$), nymphs ($N$), adults ($A$), and size host groups ($H_i$, $i\in\{1,\dots,6\}$). The three tick stages and host groups 1--5 are divided into susceptible ($S$) and infected ($I$) classes, while host group 6 is assumed incompetent and is therefore represented by a single population. To construct the model, we divided each year into seven nonoverlapping biological events (see Appendix~\ref{ap:timeline} and Table~\ref{timeline}). We then combined these events to derive the final model.

The demographic component, including host recruitment and mortality, tick survival and reproduction, seasonal feeding preferences, host-density scaling, and the event-based construction of the annual dynamics, is retained from \cite{peterson2026model}. Here, we focus on the modifications required to incorporate infection dynamics.

For each tick stage $j \in \{L,N\}$, $\beta_{i,j}$ denotes the proportion of ticks that become infected after feeding on an infected host from group $i$ with $i\in \{1,\dots,5\}$. Conversely, $\alpha_i$ represents the proportion of hosts in group $i$ that become infected following a bite from an infected nymph. For groups with multiple infectious hosts (groups 3 and 5), host infectivity is computed as a density-weighted average across species. Because host group 6 is assumed incompetent, infected deer do not contribute to pathogen transmission, and their infection dynamics are not tracked. As a result, because adults only feed on host group 6, the proportion of infected nymphs and adults is the same.

The proportion of ticks in stage $j$ feeding on host group $i$ is denoted by $p_{i,j}(t)$ and depends on the relative availability of host groups. Thus, transmission depends on both the prevalence of infection within host groups and the distribution of tick feeding across hosts.

As in \cite{peterson2026model}, differences in questing behavior are represented by $c$, the proportion of the larval and nymphal questing season spent above the leaf litter. For these stages, seasonal feeding preferences are scaled by $c$ for above-litter hosts (groups 1--4) and by $(1-c)$ for the below-litter host group (group 5). Varying $c$ therefore alters the distribution of tick-host encounters and allows us to investigate how questing behavior influences infection dynamics and measures of Lyme disease risk.

\begin{table}[h!]
    \centering
    \setlength{\tabcolsep}{5.5pt}
    \begin{tabular}{cp{7.5cm}c}
    \toprule
    \textbf{State Var.} & \textbf{Definition} & \textbf{Units}\\
    \midrule
    $L_{S,I}$ & Susceptible and infected larvae & larvae/ha\\
    $N_{S,I}$ & Susceptible and infected nymphs & nymphs/ha\\
    $A_{S,I}$ & Susceptible and infected adults & adults/ha\\
    $H_{1_{S,I}}$ & Susceptible and infected white-footed mice & mice/ha\\
    $H_{2_{S,I}}$ & Susceptible and infected chipmunks & chipmunks/ha\\
    $H_{3_{S,I}}$ & Susceptible and infected short-tailed shrews & shrews/ha\\
    $H_{4_{S,I}}$ & Susceptible and infected squirrels & squirrels/ha\\
    $H_{5_{S,I}}$ & Susceptible and infected lizards& lizards/ha\\
    $H_6$ & Deer & deer/ha\\
    \midrule
    \textbf{Parameter} & \textbf{Definition} & \textbf{Units} \\
    \midrule
    $\Lambda_{i}$  & Birth/recruitment of host group $i$ & host $i$/ha\\
    $\mu_{i}$ & Natural death of host $i$ & dimensionless\\
    $s_j$ & Natural survival of ticks in life stage $j$ & dimensionless\\
    $z_j$ & Host-finding calibration constant & dimensionless\\
    $r$ & Fecundity & larvae/deer\\
    $K$ & Tick-host ratio at which tick reproduction is at half its maximum & adults/deer\\
    $\beta_{i,j}$ & Proportion of ticks in life stage $j$ that become infected after feeding on an infected host from group $i$ & dimensionless\\
    $\alpha_i$ & Proportion of hosts from group $i$ that become infected after a bite from an infected nymph & host $i$/nymph\\
    $\hat{b}_{i,j}$ & Relative seasonal feeding potential of tick stage $j$ for host group $i$& 1/host $i$\\
    $\sigma_{i,j}$ & Host density scaling constant & dimensionless\\
    $c$ & Proportion of questing season spent above the leaf litter& dimensionless\\
    
    \bottomrule
    \end{tabular}
    \caption{State variables and parameters for Model \eqref{fullsystem} with their descriptions and units. $i \in \{1,2,3,4,5,6\}$ and $j \in \{L,N,A\}$. }
    \label{param_table}
\end{table}

\section{Baseline Demographic Dynamics}
Before analyzing infection dynamics, it is important to understand the baseline behavior of the tick-host system without pathogen transmission. The demographic component of Model~\eqref{fullsystem}, obtained by summing the susceptible and infected equations for each population, was previously analyzed in \cite{peterson2026model}. 

In that study, the long-term behavior of the tick population depended on the demographic reproductive number
\[\mathcal{R}_d= \frac{rS}{K},\]
where $S=\frac{1}{2}s_Ls_Ns_A$. The demographic reproductive number captures the balance between reproduction, survival, and the tick saturation constant. It represents the expected number of offspring per female tick over its lifetime and is derived by linearizing the system about the extinction fixed point and solving for the spectral radius of the Jacobian \citep{cushing1994net, cushing1989ebenman}. A qualitative stability analysis shows that when $\mathcal{R}_d<1$, tick populations converge to a stable extinction fixed point. For $\mathcal{R}_d=1$ the system experiences a $+1$/$-1$ bifurcation pair: a transcritical bifurcation ($\lambda_1=+1$) as the extinction fixed point exchanges stability with a positive existence fixed point, and a period doubling bifurcation ($\lambda_2=-1$) from which an unstable synchronous 2-cycle emerges, in which only larvae ($L>0, N=A=0$) or nymphs and adults ($N,A>0, L=0$) are present at a given time. This structure arises from the 2-year life cycle of \textit{I. scapularis}, which introduces a one-year delay between the absence of larvae and its effects on subsequent tick populations. For $\mathcal{R}_d>1$, the positive fixed point stabilizes while the synchronous 2-cycle remains unstable for all $\mathcal{R}_d$ values. 

\section{Qualitative Analysis}\label{modelanalysis_constant}
In this section, we analyze the infection dynamics of Model~\eqref{fullsystem} assuming that the demographic system has converged to its stable positive fixed point, so that total population sizes are constant from year to year. We derive an expression for the basic reproductive number, which provides a threshold for disease invasion, and discuss conditions for the existence of an endemic fixed point. 

\subsection{Basic Reproductive Number}\label{BRN_FP}
To determine the basic reproductive number (denoted $\mathcal{R}_0$), we analyze the infected populations of Model \eqref{fullsystem} after all demographic variables have reached equilibrium:
\begin{equation}\label{BRNsystem}
\begin{aligned}
    L_I\left(t+1\right)&=\frac{r\frac{A_{\infty}}{2}}{\frac{\frac{A_{\infty}}{2}}{H_{6_{\infty}}}+K}s_LS_{L_{\infty}}\left(\sum_{i=1}^5 p_{{i,L}_{\infty}} \beta_{i,L} \frac{H_{i_I}\left(t\right)}{H_{i_{\infty}}+\frac{\Lambda_i}{k_i}}\right)\\
    N_I\left(t+1\right)&=s_NS_{N_\infty}\left(L_I\left(t\right)+L_S\left(t\right)\left(\sum_{i=1}^5p_{{i,N}_\infty}\beta_{i,N}\frac{H_{i_I}\left(t\right)e^{-\mu_i}}{H_{i_\infty}}\right)\right)\\
    A_I\left(t+1\right)&=s_AS_{A_\infty}s_NS_{N_\infty}\left(L_I\left(t\right)+L_S\left(t\right)\left(\sum_{i=1}^5p_{{i,N}_\infty}\beta_{i,N}\frac{H_{i_I}\left(t\right)e^{-\mu_i}}{H_{i_\infty}}\right)\right)\\
    H_{i_I}\left(t+1\right)&=\left(H_{i_I}\left(t\right)+ \left(H_{i_S}(t)+\lambda_i\right)\left(1-e^{-\frac{p_{{i,N}_\infty}\alpha_{i}s_NS_{N_\infty}L_I\left(t\right)}{H_{i_\infty}}}\right)\right)e^{-\mu_i},
\end{aligned}
\end{equation}
where
\begin{align*}
    \lambda_i = \frac{\Lambda_i}{2}\left(1+e^{\frac{\mu_i}{2}}\right) \quad \text{for $i \in \{1,2,3,4\}$}, \quad
    \lambda_5 = \Lambda_5.
\end{align*}
At demographic equilibrium, both the feeding-success survival parameters ($S_j(t)$) and host contact proportions ($p_{i,j}(t)$) are constant and are denoted $S_{j_\infty}$ and $p_{{i,j}_{\infty}}$, respectively.

Because the larva-host subsystem decouples from the nymphal and adult equations, $\mathcal{R}_0$ is determined entirely by transmission processes involving larval ticks and their host interactions. This reflects a key modeling assumption: nymphs that acquire infection during the nymphal stage do not contribute to subsequent host infection as adults because group 6 is assumed to be incompetent and infection dynamics in this group are not tracked. Consequently, the adult infected class is also not involved in the transmission process here.

Applying the next-generation operator method for discrete systems \citep{allen2008basic} (see Appendix~\ref{ap:BRN}), we obtain the following expression for $\mathcal{R}_0$:
\begin{align*}
\mathcal{R}_0 &= \sqrt{
        \sum_{i=1}^5 R_iS_{L_\infty}S_{N_\infty}\frac{e^{-\mu_i}}{1-e^{-\mu_i}}},
\end{align*}
where 
\begin{align*}
    R_i&=\frac{r\frac{A_\infty}{2}}{\frac{\frac{A_\infty}{2}}{H_{6_\infty}}+K}s_Ls_N\frac{p_{i,L}p_{i,N}\alpha_i\beta_{i,L}}{H_{i_\infty}+\frac{\Lambda_i}{2}}\frac{H_{i_\infty}+\lambda_i}{H_{i_\infty}}.
\end{align*}
Here, $R_iS_{L_\infty}S_{N_\infty}$ represents the contribution of host group $i$ to pathogen transmission through a complete larva-host-larva infection cycle. The quantity $\sum_{i=1}^5R_iS_{L_\infty}S_{N_\infty}$ therefore captures the baseline larva-host-larva transmission potential of the host community over a single tick life cycle. The term $\frac{e^{-\mu_i}}{1-e^{-\mu_i}}$ represents the cumulative survival contribution of host group $i$.

By the next-generation operator method, the disease-free fixed point (where all infected populations are zero) is locally asymptotically stable whenever $\mathcal{R}_0<1$.

\subsection{Existence of an Endemic Fixed Point}\label{EE_FP}
In addition to the all-zero disease-free fixed point, we now consider the existence of a positive (endemic) fixed point. Reducing Model~\eqref{fullsystem} to a single fixed-point equation for $L_I$ and analyzing the resulting equation (see Appendix~\ref{ap:ee}) yields
\begin{align*}
    L_I=\frac{r\frac{A_{\infty}}{2}}{\frac{\frac{A_{\infty}}{2}}{H_{6_{\infty}}}+K}s_LS_{L_\infty}\left(\sum_{i=1}^5 p_{{i,L}_{\infty}} \beta_{i,L} \frac{H_{i_{\infty}}+\lambda_i}{H_{i_\infty}+\frac{\Lambda_i}{k_i}}\left(\frac{\left(1-e^{-\frac{p_{{i,N}_\infty}\alpha_{i}s_NS_{N_\infty}L_{I}}{H_{i_\infty}}}\right)e^{-\mu_i}}{1-e^{-\frac{p_{{i,N}_\infty}\alpha_{i}s_NS_{N_\infty}L_{I}}{H_{i_\infty}}}e^{-\mu_i}}\right)\right),
\end{align*}
 and shows that a unique endemic fixed point exists whenever $\mathcal{R}_0^2>1$. Biologically, this condition indicates that the combined efficiency of larva-host-larva transmission across all host groups is sufficient to sustain infection in the host-tick system. Moreover, this analysis implies that the endemic fixed point is locally asymptotically stable whenever it exists (see again Appendix~\ref{ap:ee}).

\section{Quantitative Results}
\label{simulations}
In this section, we discuss parameter estimates used to parameterize the model, detail our numerical analysis, and discuss the results. 

\subsection{Parameter Estimates}
To parameterize the model for northern and southern regions, we compiled two sets of data based on geographic location. Parameter values for both regions are presented in Table~\ref{paramvaluesnorthsouth}.

Most parameter values were obtained from published studies, and a detailed discussion can be found in \cite{peterson2026model} and \cite{peterson2026host}. 

Baseline host and tick population densities were used to calibrate the host-finding success scaling parameters $z_L$, $z_N$, and $z_A$, as well as the host mortality terms $\mu_i$. 

Host infectivity values ($\beta_{i,j}$) were based on controlled experimental studies designed to measure host infectivity rather than field prevalence \citep{brunner2008estimating, dougherty2024role, levin1996reservoir}. In the absence of data indicating otherwise, we assumed equal infectivity across tick life stages and that hosts bitten by infected nymphs always acquire infection (i.e., $\alpha_i=1$ for all $i =1,\dots,5$). Because infectivity is species-specific, we did not vary these values between regions. For brevity, parameters $\alpha_i$, $\beta_{i,L}$, and $\beta_{i,N}$ are described as infection probabilities in Table~\ref{paramvaluesnorthsouth}, though in the deterministic framework presented here they more precisely represent the proportion of infectious contacts resulting in transmission.

\begin{table}[htbp!]
\centering
\resizebox{\textwidth}{!}{%
\begin{tabular}{llrrr}
\toprule
& & \multicolumn{2}{c}{\textbf{Value [Source]}} & \\
\cmidrule(lr){3-4}
\makecell{\textbf{Baseline}\\\textbf{Densities}} & \textbf{Description} & \textbf{North} & \textbf{South} & \textbf{Units} \\
\midrule
$L(0)$ & Larvae
    & $115{,}000$ \citep{daniels2000estimating}
    & $115{,}000$ \citep{daniels2000estimating}
    & larvae/ha \\
$N(0)$ & Nymphs
    & $11{,}500$ \citep{daniels2000estimating}
    & $11{,}500$ \citep{daniels2000estimating}
    & nymph/ha \\
$A(0)$ & Adults
    & $3{,}450$ \citep{daniels2000estimating}
    & $3{,}450$ \citep{daniels2000estimating}
    & adult/ha \\
$H_{1}(0)$ & Mice
    & $40$ \citep{peterson2026host}
    & $27.9$
    & mice/ha \\
$H_{2}(0)$ & Chipmunks
    & $30$ \citep{peterson2026host}
    & $30$
    & chipmunks/ha \\
$H_{3}(0)$ & Shrews
    & $59.2$ \citep{peterson2026host}
    & $30$
    & shrews/ha \\
$H_{4}(0)$ & Squirrels
    & $6$ \citep{peterson2026host}
    & $6$
    & squirrels/ha \\
$H_{5}(0)$ & Lizards
    & $27$ \citep{peterson2026host}
    & $445.6$
    & lizard/ha \\
$H_{6}(0)$ & Deer
    & $0.46$ \citep{peterson2026host}
    & $0.40$
    & deer/ha \\

\midrule
& & \multicolumn{2}{c}{\textbf{Value [Source]}} & \\
\cmidrule(lr){3-4}
\textbf{Parameters} & \textbf{Description} & \textbf{North} & \textbf{South} & \textbf{Units} \\
\midrule
$\Lambda_1$ & Recruitment of mice
    & $300^{a}$
    & $273.5^{a}$
    & mice/ha \\
$\Lambda_2$ & Recruitment of chipmunks
    & $50$ \citep{snyder1982tamias}
    & $50$ \citep{snyder1982tamias}
    & chipmunks/ha \\
$\Lambda_3$ & Recruitment of shrews
    & $659.8^{b}$
    & $174.9^{b}$
    & shrews/ha \\
$\Lambda_4$ & Recruitment of squirrels
    & $27^{c}$
    & $27^{c}$
    & squirrels/ha \\
$\Lambda_5$ & Recruitment of lizards
    & $36.1^{d}$
    & $1278.1^{d}$
    & lizards/ha \\
$\Lambda_6$ & Recruitment of deer
    & $0.32^{e}$
    & $0.21^{e}$
    & deer/ha \\

$s_L$ & Natural survival of larvae
    & $0.43$ \citep{brunner2023off}
    & $0.43$ \citep{brunner2023off}
    & dimensionless \\
$s_N$ & Natural survival of nymphs
    & $0.55$ \citep{brunner2023off}
    & $0.55$ \citep{brunner2023off}
    & dimensionless \\
$s_A$ & Natural survival of adults
    & $0.84$ \citep{brunner2023off}
    & $0.84$ \citep{brunner2023off}
    & dimensionless \\

$\mu_1$ & Natural death of mice
    & $3.12$
    & $3.55$
    & dimensionless \\
$\mu_2$ & Natural death of chipmunks
    & $1.21$
    & $1.21$
    & dimensionless \\
$\mu_3$ & Natural death of shrews
    & $3.77$
    & $2.73$
    & dimensionless \\
$\mu_4$ & Natural death of squirrels
    & $2.36$
    & $2.36$
    & dimensionless \\
$\mu_5$ & Natural death of lizards
    & $0.85$
    & $1.35$
    & dimensionless \\
$\mu_6$ & Natural death of deer
    & $0.68$
    & $0.52$
    & dimensionless \\

$\hat{b}_{1,L}$ & Relative larval preference for mice
    & $404.4$ \citep{ginsberg2021lyme}
    & $152.0$ \citep{ginsberg2021lyme}
    & larvae/mouse \\
$\hat{b}_{2,L}$ & Relative larval preference for chipmunks
    & $245.2$ \citep{ginsberg2021lyme}
    & $220.6$
    & larvae/chipmunk \\
$\hat{b}_{3,L}$ & Relative proportion of larvae feeding on shrews
    & $483.7$ \citep{ginsberg2021lyme}
    & $1{,}693.6$ \citep{ginsberg2021lyme}
    & larvae/shrew \\
$\hat{b}_{4,L}$ & Relative proportion of larvae feeding on squirrels
    & $911.3$ \citep{main1982immature}
    & $810.8$
    & larvae/squirrel \\
$\hat{b}_{5,L}$ & Relative proportion of larvae feeding on lizards
    & $35.3$ \citep{ginsberg2021lyme}
    & $25.7$ \citep{ginsberg2021lyme}
    & larvae/lizard \\
$\hat{b}_{6,L}$ & Relative proportion of larvae feeding on deer
    & $1{,}0817.1$ \citep{piesman1979role}
    & $962.3$
    & larvae/host 6 \\

$\hat{b}_{1,N}$ & Relative proportion of nymphs feeding on mice
    & $20.4$ \citep{ginsberg2021lyme}
    & $15.2$ \citep{ginsberg2021lyme}
    & nymphs/mouse \\
$\hat{b}_{2,N}$ & Relative proportion of nymphs feeding on chipmunks
    & $188.4$ \citep{ginsberg2021lyme}
    & $167.6$
    & nymphs/chipmunk \\
$\hat{b}_{3,N}$ & Relative proportion of nymphs feeding on shrews
    & $1.9$ \citep{ginsberg2021lyme}
    & $1.4$
    & nymphs/shrew \\
$\hat{b}_{4,N}$ & Relative proportion of nymphs feeding on squirrels
    & $1{,}527.8$ \citep{main1982immature}
    & $1{,}359.3$
    & nymphs/squirrel \\
$\hat{b}_{5,N}$ & Relative proportion of nymphs feeding on lizards
    & $0.4$ \citep{ginsberg2021lyme}
    & $3.0$ \citep{ginsberg2021lyme}
    & nymphs/lizard \\
$\hat{b}_{6,N}$ & Relative proportion of nymphs feeding on deer
    & $4{,}132.2$ \citep{piesman1979role}
    & $3{,}676.5$
    & nymphs/deer \\

$b_{6,A}$ & Relative proportion of adults feeding on deer
    & $1{,}037.5$ \citep{piesman1979role}
    & $1{,}037.5$ \citep{piesman1979role}
    & adults/deer \\

\makecell[l]{$\sigma_{1,N}$\\\vspace{0.1cm}}
& \makecell[l]{Scaling constant for preference of nymphs\\
(versus larval) for host group 1}
& ---
& \makecell[r]{$1.24$\\\vspace{0.1cm}}
& \makecell[r]{dimensionless\\\vspace{0.1cm}} \\

\makecell[l]{$\sigma_{3,N}$\\\vspace{0.1cm}}
& \makecell[l]{Scaling constant for preference of nymphs\\
(versus larval) for host group 3}
& \makecell[r]{$3.59$\\\vspace{0.1cm}}
& \makecell[r]{$0.90$\\\vspace{0.1cm}}
& \makecell[r]{dimensionless\\\vspace{0.1cm}} \\

\makecell[l]{$\sigma_{5,N}$\\\vspace{0.1cm}}
& \makecell[l]{Scaling constant for preference of nymphs\\
(versus larval) for host group 5}
& \makecell[r]{$1.07$\\\vspace{0.1cm}}
& \makecell[r]{$1.19$\\\vspace{0.1cm}}
& \makecell[r]{dimensionless\\\vspace{0.1cm}} \\

\makecell[l]{$\sigma_{6,N}$\\\vspace{0.1cm}}
& \makecell[l]{Scaling constant for preference of nymphs\\
(versus larval) for host group 6}
& \makecell[r]{$1.25$\\\vspace{0.1cm}}
& \makecell[r]{$1.00$\\\vspace{0.1cm}}
& \makecell[r]{dimensionless\\\vspace{0.1cm}} \\

\makecell[l]{$\sigma_{6,A}$\\\vspace{0.1cm}}
& \makecell[l]{Scaling constant for preference of adults\\
(versus larval) for host group 6}
& \makecell[r]{$1.08$\\\vspace{0.1cm}}
& \makecell[r]{$1.00$\\\vspace{0.1cm}}
& \makecell[r]{dimensionless\\\vspace{0.1cm}} \\

$r$ & Fecundity of adult ticks
    & $2{,}520{,}000$ \citep{baker2024feeding, piesman1979role}
    & $2{,}520{,}000$
    & larvae/host 6 \\

\makecell[l]{$K$\\\vspace{0.1cm}}
& \makecell[l]{Tick-to-host ratio at which tick reproduction is at\\
half its maximum}
& \makecell[r]{$5$\\\vspace{0.1cm}}
& \makecell[r]{$5$\\\vspace{0.1cm}}
& \makecell[r]{adult/host 6\\\vspace{0.1cm}} \\

$z_L$ & Larval host-finding success calibration constant
    & $1.24$
    & $1.45$
    & dimensionless \\
$z_N$ & Nymphal host-finding success calibration constant
    & $1.30$
    & $2.58$
    & dimensionless \\
$z_A$ & Adult host-finding success calibration constant
    & $1.79$
    & $2.23$
    & dimensionless \\
$c$ & Proportion of questing season spent above the leaf litter
    & $0.18$ \citep{arsnoe2019nymphal}
    & $0.02$ \citep{arsnoe2019nymphal}
    & dimensionless \\
    $\alpha_1$ & Infection probability from nymphs to mice  & $1$ \citep{sze2012borrelia, brunner2008estimating}
    & $1$ \citep{sze2012borrelia, brunner2008estimating} & mice/nymph \\
    $\alpha_2$ & Infection probability from nymphs to chipmunks  & $1$ & $1$ & chipmunks/nymph \\
    $\alpha_3$ & Infection probability from nymphs to shrews  & $1$ & $1$  & shrews/nymph \\
    $\alpha_4$ & Infection probability from nymphs to squirrels  & $1$ & $1$ & squirrels/nymph \\
    $\alpha_5$ & Infection probability from nymphs to lizards  & $1$ & $1$ & lizards/nymph \\
    $\beta_{1,L}$ & Infection probability from mice to larvae & $0.91$ \citep{brunner2008estimating} & $0.91$ \citep{brunner2008estimating} & dimensionless \\ 
    $\beta_{2,L}$ & Infection probability from chipmunks to larvae & $0.57$ \citep{brunner2008estimating} & $0.57$ \citep{brunner2008estimating} & dimensionless \\
    $\beta_{3,L}$ & Infection probability from shrews to larvae & $0.56$ \citep{brunner2008estimating, dougherty2024role} & $0.56$ \citep{brunner2008estimating, dougherty2024role} & dimensionless \\
    $\beta_{4,L}$ & Infection probability from squirrels to larvae & $0.06$ \citep{brunner2008estimating} & $0.06$ \citep{brunner2008estimating} & dimensionless \\
    $\beta_{5,L}$ & Infection probability from lizards to larvae & $0.02$ \citep{levin1996reservoir} & $0.02$ \citep{levin1996reservoir} & dimensionless \\
    $\beta_{1,N}$ & Infection probability from mice to nymphs & $0.91$ \citep{brunner2008estimating} & $0.91$ \citep{brunner2008estimating} & dimensionless \\ 
    $\beta_{2,N}$ & Infection probability from chipmunks to nymphs & $0.57$ \citep{brunner2008estimating} & $0.57$ \citep{brunner2008estimating} & dimensionless \\
    $\beta_{3,N}$ & Infection probability from shrews to nymphs & $0.56$ \citep{brunner2008estimating, dougherty2024role} & $0.56$ \citep{brunner2008estimating, dougherty2024role} & dimensionless \\
    $\beta_{4,N}$ & Infection probability from squirrels to nymphs & $0.06$ \citep{brunner2008estimating} & $0.06$ \citep{brunner2008estimating} & dimensionless \\
    $\beta_{5,N}$ & Infection probability from lizards to nymphs & $0.02$ \citep{levin1996reservoir} & $0.02$ \citep{levin1996reservoir} & dimensionless \\

\bottomrule
\end{tabular}
}

\caption{Baseline population densities and parameters estimated from literature for the northeast (North) and southeast (South) for Model~\eqref{fullsystem} with their descriptions, estimated values, and units. Parameters without cited sources were calculated as described in \cite{peterson2026model}.\\
$^{a}$\citep{morris1986proximate, mehlhop1978population, chitty1941territorial}\\
$^{b}$\citep{innes1990high, whitaker1998mammals, merritt1981clethrionomys, innes1984life}\\
$^{c}$\citep{layne1952genitale, barkalow1970vital, kemp1970dynamics}\\
$^{d}$\citep{COSEWIC2007, seburn1990population, tinkle1972sceloporus}\\
$^{e}$\citep{verme1984physiology, lay1942ecology, johnson2016habitat}}
\label{paramvaluesnorthsouth}

\end{table}

The equilibria of Model~\eqref{fullsystem} were computed numerically in MATLAB\textsuperscript{\textregistered} with \verb|fsolve|, and the resulting solutions are presented in the following sections. Multiple initial conditions spanning the feasible ranges of tick and host populations were used to locate biologically meaningful fixed points. For each equilibrium, local stability was determined from the eigenvalues of the Jacobian matrix, which was evaluated numerically at the corresponding solution. The same numerical approach was applied to the second-generation system to identify 2-cycles and evaluate their stability.

\subsection{Demographic and Behavioral Influence on the Basic Reproductive Number}\label{quantitative_analysis}

In this section, we investigate how key demographic and behavioral parameters influence pathogen transmission through their effects on the basic reproductive number $\mathcal{R}_0$. Since $\mathcal{R}_0$ is defined only for stable demographic dynamics, this quantity is computed only when the positive existence fixed point is stable (i.e., when $\mathcal{R}_d>1$).

\begin{figure*}[ht!]
\centering
\includegraphics[width=\linewidth]{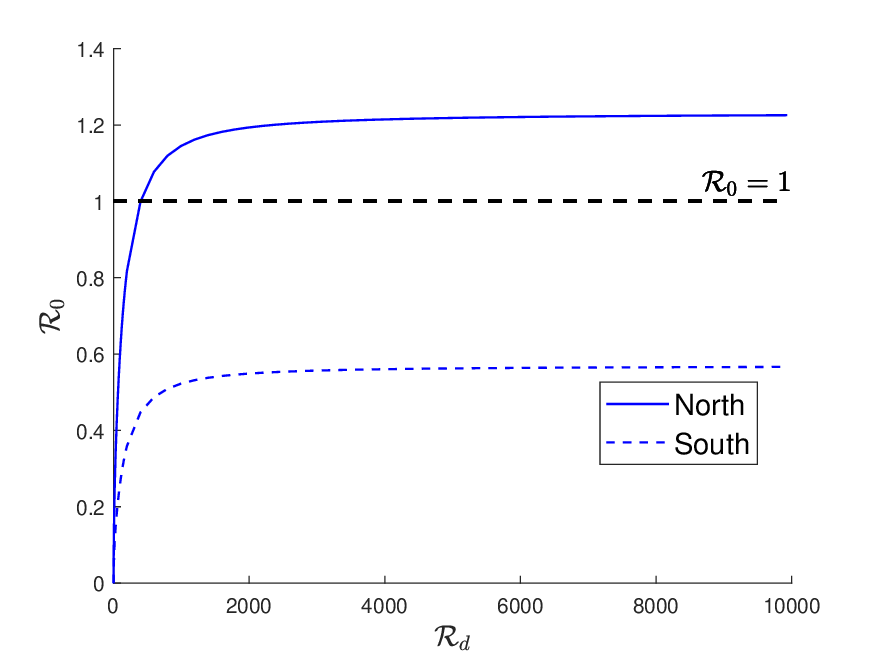}
    \caption{Plot of $\mathcal{R}_0$ as a function of the demographic reproductive number $\mathcal{R}_d=\frac{rS}{K}$ under the northern (solid) and southern (dashed) parameter sets from Table~\ref{paramvaluesnorthsouth}, except $r$, which is varied to change $\mathcal{R}_d$. In the figure, blue lines indicate $\mathcal{R}_0$ values evaluated at the coexistence fixed point, while the dashed black line represents the threshold $\mathcal{R}_0=1$}
    \label{R0_rSK}
\end{figure*}
\begin{figure*}[ht!]
\centering
    \includegraphics[width=\linewidth]{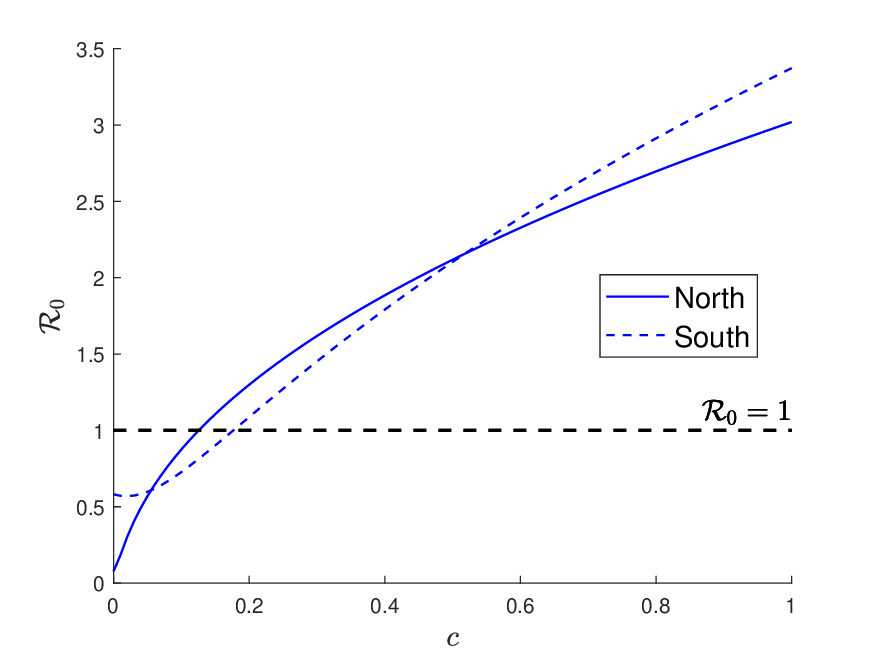}
    \caption{Diagram of $\mathcal{R}_0$ with respect to the questing parameter $c$ under the northern (solid) and southern (dashed) parameter sets from Table~\ref{paramvaluesnorthsouth}. Blue lines indicate $\mathcal{R}_0$ values evaluated at the coexistence fixed point, while the dashed black line represents the threshold $\mathcal{R}_0=1$.}
    \label{R0_c}
\end{figure*}

Figure~\ref{R0_rSK} illustrates how the basic reproductive number varies with the demographic reproductive number $\mathcal{R}_d=\frac{rS}{K}$ under the northern and southern parameter sets. In the figure, the blue curves represent $\mathcal{R}_0$ while the dashed black line represents the threshold $\mathcal{R}_0=1$. 

Under the northern parameter set, $\mathcal{R}_0$ exceeds one for nearly all of the values of $\mathcal{R}_d$ considered. Our estimated demographic reproductive number, $\mathcal{R}_d=50{,}000$, lies well within this range, indicating that pathogen invasion is expected under northern conditions. In contrast, under the southern parameter set, $\mathcal{R}_0$ remains below one throughout the parameter range, so the same estimated value of $\mathcal{R}_d$ predicts that infection cannot be maintained despite demographic persistence of the tick population. In both regions, $\mathcal{R}_0$ approaches an asymptotic value as $\mathcal{R}_d$ increases. This saturation arises from the ratio-dependent host-finding success: as tick abundance increases relative to the available hosts, host-finding success becomes limited, so further increases in $\mathcal{R}_d$ produce smaller increases in transmission. Host demographic turnover also influences $\mathcal{R}_0$ by limiting opportunities for pathogen transmission. 

If we consider density-independent host-finding success (that is, taking $S_j(t)$ to be a constant value for all $t$), a different behavior is observed (see Figure~\ref{R0_rSK_constant} in Appendix~\ref{ap:BRN_constandsuccess}). In that case, $\mathcal{R}_0$ increases much higher and faster than in the ratio-dependent case, yielding higher estimates for $\mathcal{R}_0$ under both parameter regions.

Similarly, Figure~\ref{R0_c} shows the dependence of $\mathcal{R}_0$ on the questing parameter $c$, which determines the proportion of the questing season ticks spend above the leaf litter. As before, blue lines correspond to $\mathcal{R}_0$ evaluated at the stable fixed point and dashed black lines represent the threshold $\mathcal{R}_0=1$.

For small values of the questing parameter ($0<c<0.06$), the southern $\mathcal{R}_0$ exceeds the northern value, although both remain below one. Between approximately $0.06<c<0.52$, $\mathcal{R}_0$ is larger under the northern parameter set, with both curves crossing the line $\mathcal{R}_0=1$ in this region. For larger values of $c$, the southern $\mathcal{R}_0$ again exceeds the northern value. Comparing baseline parameter values ($c=0.18$ for the north and $c=0.02$ for the south), the model predicts $\mathcal{R}_0>1$ under northern conditions and $\mathcal{R}_0<1$ under southern conditions, indicating sustained transmission only under the northern baseline scenario. 

\subsection{Effects of Questing Behavior and Transmission Parameters on Nymphal Infection Metrics}\label{DIN_NIP}
Two commonly used metrics for assessing human Lyme disease risk are nymphal infection prevalence (NIP) and density of infected nymphs (DIN) \citep{logiudice2003ecology, lou2014impact}. NIP measures the proportion of questing nymphs that are infected with \textit{B. burgdorferi}, while DIN represents the density of infected nymphs in the environment. Although related, these quantities can produce contradicting results, highlighting the importance of evaluating them together when assessing disease risk \citep{peterson2026host}. To evaluate how questing behavior and host community composition affect these quantities, we compute NIP and DIN as functions of the questing parameter $c$. We also simulate infected host communities as functions of infection parameters. 

\begin{figure*}[ht!]
\centering
\begin{subfigure}[t]{0.5\textwidth}
    \centering
    \includegraphics[width=\linewidth]{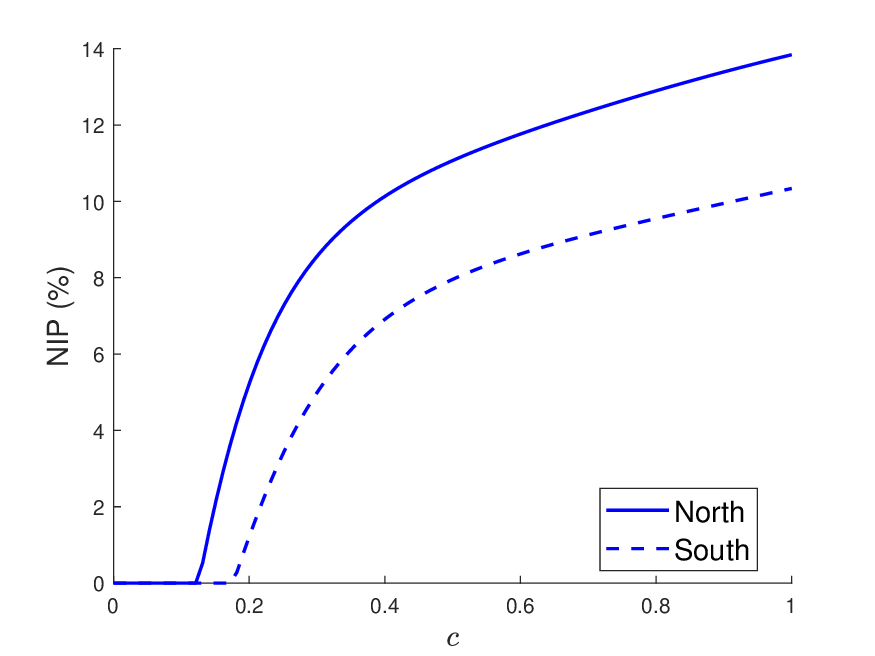}
    \caption{}
    \label{fig:4a}
\end{subfigure}%
\begin{subfigure}[t]{0.5\textwidth}
    \centering
    \includegraphics[width=\linewidth]{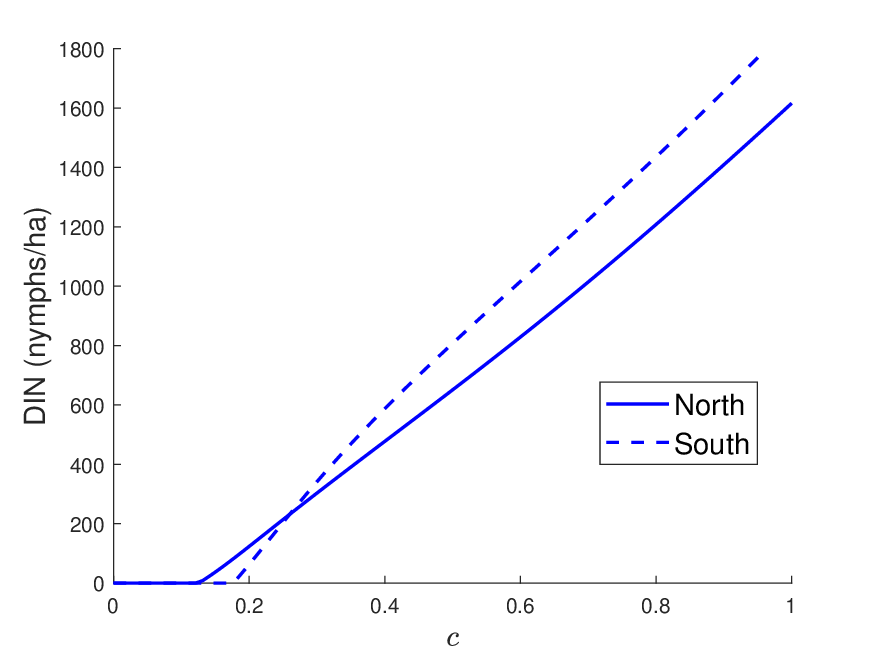}
    \caption{}
    \label{fig:4b}
\end{subfigure}%
\caption{Nymphal infection prevalence (NIP) and density of infected nymphs (DIN) as functions of the questing parameter $c$ under the northern (solid) and southern (dashed) parameter sets from Table~\ref{paramvaluesnorthsouth}.}
    \label{NIP_DIN_c}
\end{figure*}

\begin{figure*}[ht!]
\centering
\begin{subfigure}[t]{0.5\textwidth}
    \centering
    \includegraphics[width=\linewidth]{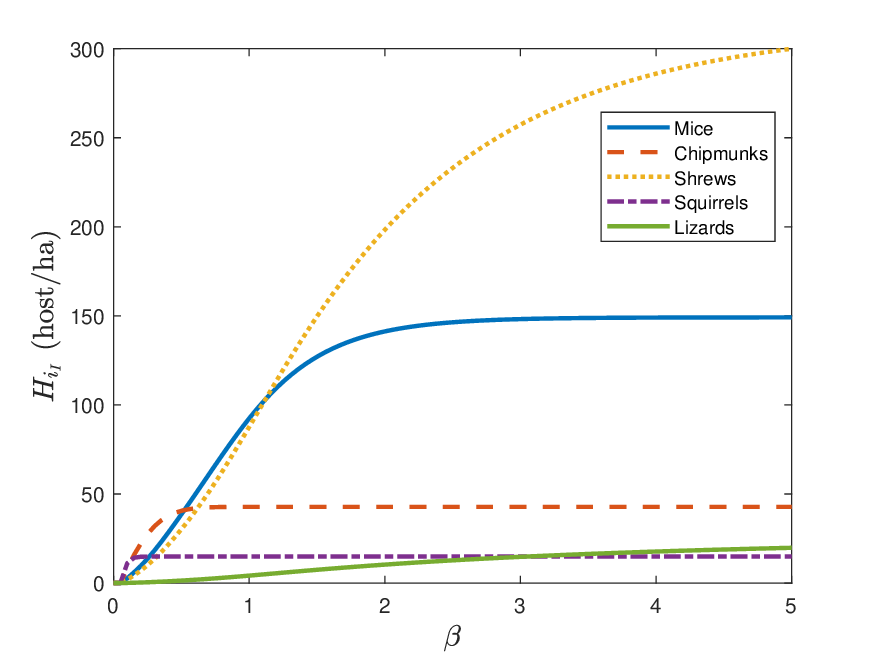}
    \caption{}
    \label{fig:5a}
\end{subfigure}%
\begin{subfigure}[t]{0.5\textwidth}
    \centering
    \includegraphics[width=\linewidth]{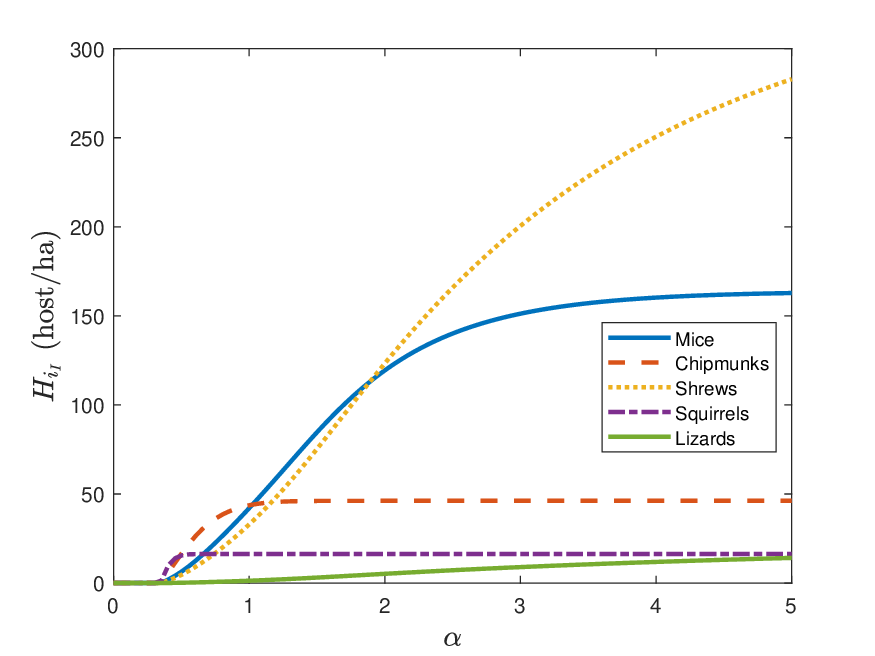}
    \caption{}
    \label{fig:5b}
\end{subfigure}%
    \caption{Densities of infected hosts ($H_{i_I}$ for $i \in \{1,2,3,4,5\}$) as the host-to-tick transmission probabilities $\beta_{i}$ (subfigure~\ref{fig:5a}) and the tick-to-host transmission probabilities $\alpha_i$ (subfigure~\ref{fig:5b}) are uniformly scaled by the parameters $\beta$ and $\alpha$ respectively under the northern parameter set from Table~\ref{paramvaluesnorthsouth}.}
    \label{HI_beta_alpha}
\end{figure*}

Figure~\ref{NIP_DIN_c} shows how the questing parameter $c$ influences NIP (Subfigure~\ref{fig:4a}) and DIN (Subfigure~\ref{fig:4b}) under both the northern  and southern parameter sets given in Table~\ref{paramvaluesnorthsouth}.

No positive endemic fixed point exists for $c<0.13$ under the northern parameter set or for $c<0.18$ under the southern parameter set. Over these intervals $\mathcal{R}_0<1$, so the disease-free fixed point is locally asymptotically stable, and both NIP and DIN are zero. For values of $c$ where an endemic fixed point exists, northern NIP exceeds southern NIP for all values of $c$ considered. For DIN, however, northern values exceed southern values for approximately $0.13<c<0.26$, after which southern DIN is larger. Thus, the two metrics do not always yield the same comparison between regions; for $0.26<c<1$, NIP is higher in the north, while DIN is higher in the south.

Figure~\ref{HI_beta_alpha} shows the densities of infected hosts as the host-to-larvae transmission probabilities $\beta_{i,L}$ and nymph-to-host transmission probabilities $\alpha_i$ are uniformly scaled by parameters $\beta$ and $\alpha$, respectively. In both cases, infected host densities increase initially, then level off at fixed values. This saturation reflects the influence of host mortality and turnover, which limit the duration of infection within the host population. As a result, even for large transmission probabilities, infected host densities remain bounded. 

This constraint also limits the impact of transmission dynamics on model behavior: for the parameter ranges considered here, the addition of disease dynamics does not produce oscillatory behavior or additional period-doubling bifurcations.

\section{Discussion and Conclusions}
\label{discussion}

In this work, we developed a new mathematical model for representing the dynamics of sylvatic host-vector transmission of {\it Borrelia burgdorferi}. We created a system of difference equations by layering a chronologically ordered sequence of demographic and epidemiological events involving three \textit{Ixodes scapularis} life stages and six host groups in order to capture the complex tick life cycle while preserving some level of analytical tractability. This system builds on our prior work, incorporating tick questing behavior as well as ratio-dependent host-finding success \citep{peterson2026model}. Our analysis shows how these factors impact transmission dynamics.

We first analyzed disease dynamics by deriving an expression for the basic reproductive number, $\mathcal{R}_0$. The resulting expression shows that transmission depends on host-finding success, captured by $S_{L_{\infty}}$ and $S_{N_{\infty}}$, both of which are strictly less than one when all populations are positive. These factors reduce the number of successful contacts and therefore limit transmission.

We then examined the role of host community composition and the demographic reproductive number $\mathcal{R}_d$, which represents the expected number of offspring produced by an individual over its lifetime. Across the range of $\mathcal{R}_d$ examined at the respective baseline questing parameters, $\mathcal{R}_0$ was larger under the northern parameter set than under the southern one, indicating that the probability of transmission is higher in the north than in the south with increased survival and reproductive success. In addition, $\mathcal{R}_0$ exceeded one for much of the $\mathcal{R}_d$ range in the north, implying that transmission can be sustained, except when $\mathcal{R}_d$ is extremely low and tick populations approach extinction. In contrast, $\mathcal{R}_0$ remained below one in southern systems across all values of $\mathcal{R}_d$ considered, suggesting that the pathogen cannot persist. These results are consistent with the observed geographic distribution of Lyme disease cases and infected ticks \citep{eisen2016county, foster2023prevalence, foster2024density,cdc_lyme_case_map}. Importantly, these qualitative differences between northern and southern systems emerged from parameter values obtained independently from the literature rather than from fitting or tuning the model to reproduce observed geographic patterns. 

Comparing these results with those obtained under density-independent success (Appendix~\ref{ap:BRN_constandsuccess}), the model predicts much larger values for $\mathcal{R}_0$ in both regions, including values above one in the south. In contrast, the ratio-dependent feeding formulation accounts for competition among questing ticks for available hosts, producing a saturating tick population and thus smaller $\mathcal{R}_0$ values. These results suggest that assuming density-independent host-finding success may overestimate transmission potential.

We also explored how questing behavior and host community composition influence transmission. When the questing parameter $c$ was varied over its full range, our simulations showed that southern host communities are capable of supporting values of $\mathcal{R}_0$ comparable to, and often exceeding, those obtained under northern parameter sets. For small values of $c$, corresponding to ticks spending most of the questing season within the leaf litter, southern parameter sets produce larger values of $\mathcal{R}_0$ reflecting the abundance of below-litter hosts that promote successful feeding. As $c$ increases, northern systems exhibit larger $\mathcal{R}_0$ values because a greater proportion of ticks encounter highly competent mammalian reservoirs, whereas a larger fraction of southern ticks continue feeding on incompetent reptiles. For larger values of $c$, however, the greater overall host abundance in the southern community produces more feeding opportunities and transmission events, resulting in larger values of $\mathcal{R}_0$ than under northern conditions. Comparing baseline scenarios ($c=0.18$ in the north and $c=0.02$ in the south), we see that under northern baseline conditions, $\mathcal{R}_0>1$, whereas southern baseline conditions yield $\mathcal{R}_0<1$. 

The same pattern is reflected in the simulation of NIP and DIN. Although southern parameter sets produce larger values of these quantities when questing behavior varies, the observed baseline questing behaviors show that both NIP and DIN are lower in the south. 

Taken together, these results show that, in our model, greater host diversity (like that in the southeastern U.S.) alone is not sufficient to suppress transmission under all questing behaviors. Instead, the combination of tick questing behavior and host community composition likely drives the geographic difference in Lyme disease prevalence. This result could have important consequences as ticks from northern regions (who tend to quest higher) continue to expand southward \cite{eisen2023changes, eisen2024evaluation}. These conclusions highlight the importance of considering vector behavior alongside host community composition when assessing transmission risk.

Finally, we examined how transmission parameters influence infected host densities (Figure~\ref{HI_beta_alpha}). Increasing transmission probabilities leads to an initial rise in infected host densities, followed by a saturation. This behavior is driven by host mortality and turnover, which limit how long individuals remain infected. As a result, even large increases in transmission efficiency do not produce unbounded growth in infection levels. This explains why the infection subsystem does not generate overshoot-and-compensate behaviors that lead to oscillations: for the parameter ranges considered, demographic turnover prevents the infection subsystem from exhibiting oscillatory dynamics.

Comparing our results to published field studies reveals both similarities and discrepancies. At the baseline questing behavior, the model reproduces the qualitative pattern of low nymphal infection prevalence and density of infected nymphs observed throughout the southeastern U.S., where reported values are often near zero \cite{foster2023prevalence, diuk2012human}. In contrast, predicted NIP and DIN values for northern ecosystems are substantially lower than those reported from the northeastern and upper midwestern U.S., where infection prevalence commonly ranges from 15--39\% \cite{logiudice2003ecology, hersh2014co, hanincova2006epidemic, huang2019high, conte2019effects, larson2017ecology}. These differences may reflect ecological variation not fully captured by the model, including assumptions about the timing of host infectiousness and demographic processes that influence opportunities for pathogen transmission. Thus, comparisons with field observations provide useful context for interpreting model predictions and identifying biological processes that may be important to incorporate in future work. 

While this study provides insight into how questing behavior and host structure influence disease dynamics, several limitations should be noted. The model relies on parameter values drawn from studies conducted across different locations and ecological contexts. Consequently, the results presented here should be interpreted with the model assumptions and parameterization in mind. In addition, changes in questing behavior likely have consequences on tick survival, which we did not account for in this study. Specifically, increased time spent above the leaf litter would likely elevate desiccation risk, particularly in southern regions \citep{ginsberg2017environmental,leal2020questing, eisen2016linkages}. As a simplifying assumption, we assumed that hosts infected at the beginning of the questing season do not become infectious until the end of the season, which likely underestimates NIP and DIN. Furthermore, the model assumes that different tick life stages do not quest together. In reality, there is likely overlap, particularly between larval and nymphal stages in the southern U.S., where longer questing seasons lead to simultaneous questing \citep{eisen2025seasonal}. This overlap may further contribute to reduced infection prevalence in the south, as larvae may be less likely to feed on hosts that were recently infected by nymphs.

An important extension of the current framework is to allow behavioral and demographic parameters to vary with environmental conditions rather than remain fixed. In particular, temperature, humidity, and soil conditions are known to influence questing behavior, survival, and host-seeking activity in \textit{I. scapularis} populations \citep{vail2002vertical,burtis2016impact,mcclure2019climate,cook2014effects}. Incorporating these climatic drivers into the model would provide a tool for investigating how changing environmental conditions may alter future Lyme disease risk.

Overall, these results demonstrate that both tick-host interactions and questing behavior strongly influence pathogen invasion. Ratio-dependent host-finding success produces saturating behavior in $\mathcal{R}_0$ and NIP over the parameter ranges examined. Differences in questing behavior further amplify these effects: increased time spent above the leaf litter increases transmission by redistributing host encounters towards more competent hosts. Notably, the model reproduces the observed geographic contrast in Lyme disease dynamics without being explicitly calibrated to that pattern. This qualitative agreement provides support for the plausibility of the framework and its ability to capture some of the ecological mechanisms underlying the geographic variation in Lyme disease risk. These results further demonstrate that differences in vector behavior and host community structure, rather than biodiversity alone, may help explain the pronounced difference in Lyme disease prevalence between the northeastern and southeastern United States.

\backmatter

\bmhead{Supplementary information}

Not applicable

\bmhead{Acknowledgements}

\section*{Declarations}

\bmhead{Funding}
This work was supported by the US National Science Foundation under Grant No. DMS-2230790.

\bmhead{Competing Interests}
The authors declare that they have no known competing financial interests or personal relationships that could have appeared to influence the work reported in this paper.

\bmhead{Ethics approval and consent to participate}
Not applicable.

\bmhead{Consent for publication}
Not applicable.

\bmhead{Data and Code availability}
The code used for numerical simulations and to generate the figures in this paper is available at \url{https://github.com/carlipeterson/jmb-lymedisease}.

\bmhead{Materials availability}
Not applicable.

\bmhead{Author Contributions}
Conceptualization, C.P., E.G., C.M.K; 
methodology, C.P., C.M.K.; 
software, C.P.; 
formal analysis, C.P., C.M.K.; 
investigation, C.P., C.M.K.; 
writing---original draft preparation, C.P.; 
writing---review and editing, C.P., C.M.K.; 
visualization, C.P.; 
supervision, E.G., C.M.K.; funding acquisition, C.M.K.

\begin{appendices}

\section{Timeline of Events}\label{ap:timeline}
The following list is a sequential timeline of the tick-host life cycle with infection dynamics split into seven time steps. 
\begin{itemize}
    \item[1.] Nymphs die (from climate-related factors ($s_N$) and from failure to find a host):
    \begin{align*}
        N_S\left(t+\frac{1}{7}\right)&=s_NN_S(t)\left(1-e^{-\frac{z_N}{L(t)}\sum_{i=1}^6 b_{i,N}\sigma_{i,N}H_i(t)}\right)\\
        N_I\left(t+\frac{1}{7}\right)&=s_NN_I\left(t\right)\left(1-e^{-\frac{z_N}{L(t)}\sum_{i=1}^6 b_{i,N}\sigma_{i,N}H_i(t)}\right)
    \end{align*}
    \item[2.] Nymphs feed and may become infected or may infect hosts:
    \begin{align*}
        H_{i_S}\left(t+\frac{2}{7}\right)&=H_{i_S}\left(t+\frac{1}{7}\right) e^{-\frac{p_{i,N}\left(t+\frac{1}{7}\right)\alpha_iN_I\left(t+\frac{1}{7}\right)}{H_i\left(t+\frac{1}{7}\right)}}\quad \text{for $i=1,2,3,4,5$}\\
        H_{i_I}\left(t+\frac{2}{7}\right)&=H_{i_I}\left(t+\frac{1}{7}\right)+H_{i_S}\left(t+\frac{1}{7}\right) \left(1-e^{-\frac{p_{i,N}\left(t+\frac{1}{7}\right)\alpha_iN_I\left(t+\frac{1}{7}\right)}{H_i\left(t+\frac{1}{7}\right)}}\right)\quad \text{for $i=1,2,3,4,5$}\\~\\
        N_S\left(t+\frac{2}{7}\right)&=N_S\left(t+\frac{1}{7}\right)\left(1-\sum_{i=1}^5 p_{i,N}\left(t+\frac{1}{7}\right) \beta_{i,N} \frac{H_{i_I}\left(t+\frac{1}{7}\right)}{H_i\left(t+\frac{1}{7}\right)}\right)\\
        N_I\left(t+\frac{2}{7}\right)&=N_I\left(t+\frac{1}{7}\right)+N_S\left(t+\frac{1}{7}\right)\sum_{i=1}^5 p_{i,N}\left(t+\frac{1}{7}\right) \beta_{i,N} \frac{H_{i_I}\left(t+\frac{1}{7}\right)}{H_i\left(t+\frac{1}{7}\right)}
    \end{align*}
    \item[3.] Half of the year's hosts are born for host groups 1--4, larvae hatch, and nymphs molt to adults:
    \begin{align*}
        H_{i_S}\left(t+\frac{3}{7}\right)&=H_{i_S}\left(t+\frac{2}{7}\right)+\frac{\Lambda_i}{2} \quad \text{for $i=1,2,3,4$}\\
         H_{5_S}\left(t+\frac{3}{7}\right)&=H_5\left(t+\frac{2}{7}\right)+\Lambda_5\\
        H_6\left(t+\frac{3}{7}\right)&=H_6\left(t+\frac{2}{7}\right)+\Lambda_6\\~\\
           L_S\left(t+\frac{3}{7}\right) &= \frac{r\frac{A\left(t-\frac{1}{7}\right)}{2}}{\frac{\frac{A\left(t-\frac{1}{7}\right)}{2}}{H_6\left(t-\frac{1}{7}\right)}+K}\\~\\
           A_S\left(t+\frac{3}{7}\right)&=N_S\left(t+\frac{2}{7}\right)\\
            A_I\left(t+\frac{3}{7}\right)&=N_I\left(t+\frac{2}{7}\right)\\~\\
            N_S\left(t+\frac{3}{7}\right)&=0\\
            N_I\left(t+\frac{3}{7}\right)&=0\\
    \end{align*}
    \item[4.] Larvae die (from climate-related factors ($s_L$) and from failure to find a host and feed), and half of the year's host deaths occur:
    \begin{align*}
         L_S\left(t+\frac{4}{7}\right)&=s_LL_S\left(t+\frac{3}{7}\right)\left(1-e^{-\frac{z_L}{B\left(t+\frac{3}{7}\right)}\sum_{i=1}^6 b_{i,L}H_i\left(t+\frac{3}{7}\right)}\right)\\~\\
        H_{i_S}\left(t+\frac{4}{7}\right) &= H_{i_S}\left(t+\frac{3}{7}\right)e^{-\frac{\mu_{i}}{2}} \quad \text{for $i=1,2,3,4,5$}\\
        H_{i_I}\left(t+\frac{4}{7}\right) &= H_{i_I}\left(t+\frac{3}{7}\right)e^{-\frac{\mu_{i}}{2}} \quad \text{for $i=1,2,3,4,5$}\\
        H_6\left(t+\frac{4}{7}\right)&= H_6\left(t+\frac{3}{7}\right)e^{-\frac{\mu_{6}}{2}}
    \end{align*}
    \item[5.] Larvae feed and potentially become infected, and adults die (from climate-related factors ($s_A$) and from failure to find a host and feed):
    \begin{align*}
        L_S\left(t+\frac{5}{7}\right)&=L_S\left(t+\frac{4}{7}\right)\left(1-\sum_{i=1}^5 p_{i,L}\left(t+\frac{4}{7}\right) \beta_{i,L} \frac{H_{i_I}\left(t+\frac{4}{7}\right)}{H_i\left(t+\frac{4}{7}\right)}\right)\\
        L_I\left(t+\frac{5}{7}\right)&=L_I\left(t+\frac{4}{7}\right)+L_S\left(t+\frac{4}{7}\right)\sum_{i=1}^5 p_{i,L}\left(t+\frac{4}{7}\right) \beta_{i,L} \frac{H_{i_I}\left(t+\frac{4}{7}\right)}{H_i\left(t+\frac{4}{7}\right)}\\~\\
        A_S\left(t+\frac{5}{7}\right)&=s_AA_S\left(t+\frac{4}{7}\right)\left(1-e^{-\frac{z_A}{N\left(t+\frac{4}{7}\right)} b_{6,A}\sigma_{6,A}H_6\left(t+\frac{4}{7}\right)}\right)\\
        A_I\left(t+\frac{5}{7}\right)&=s_AA_I\left(t+\frac{4}{7}\right)\left(1-e^{-\frac{z_A}{N\left(t+\frac{4}{7}\right)}b_{6,A}\sigma_{6,A}H_6\left(t+\frac{4}{7}\right)}\right)
    \end{align*}
    \item[6.] Larvae molt to nymphs, female adults feed and lay eggs, and the second half of the year's hosts are born for host groups 1--4:
    \begin{align*}
        N_S\left(t+\frac{6}{7}\right)&=L_S\left(t+\frac{5}{7}\right)\\
        N_I\left(t+\frac{6}{7}\right)&=L_I\left(t+\frac{5}{7}\right)\\~\\
        L_S\left(t+\frac{6}{7}\right)&=0\\
        L_I\left(t+\frac{6}{7}\right)&=0\\~\\
        H_{i_S}\left(t+\frac{6}{7}\right)&=H_{i_S}\left(t+\frac{5}{7}\right)+\frac{\Lambda_i}{2} \quad \text{for $i=1,2,3,4$}\\~\\
        A_S\left(t+\frac{6}{7}\right)&=A_S\left(t+\frac{5}{7}\right)\\
        A_I\left(t+\frac{6}{7}\right)&=A_I\left(t+\frac{5}{7}\right)
    \end{align*}
    \item[7.] All adults die (after reproducing), and the second half of the year's host deaths occur:
    \begin{align*}
        A_S\left(t+1\right)&=0\\
        A_I\left(t+1\right)&=0\\~\\
        H_{i_S}\left(t+1\right) &= H_{i_S}\left(t+\frac{6}{7}\right)e^{-\frac{\mu_{i}}{2}} \quad \text{for $i=1,2,3,4,5$}\\
        H_{i_I}\left(t+1\right) &= H_{i_I}\left(t+\frac{6}{7}\right)e^{-\frac{\mu_{i}}{2}}\quad \text{for $i=1,2,3,4,5$}\\
        H_6\left(t+1\right)&= H_6\left(t+\frac{6}{7}\right)e^{-\frac{\mu_{6}}{2}}
    \end{align*}
\end{itemize}

\newpage
\begin{sidewaystable}
\setlength{\tabcolsep}{3pt}
\renewcommand{\arraystretch}{1.1}
{\begin{tabular}{l|l|ccccccc}
\toprule
 & &
$t+\frac{1}{7}$ &
$t+\frac{2}{7}$ &
$t+\frac{3}{7}$ &
$t+\frac{4}{7}$ &
$t+\frac{5}{7}$ &
$t+\frac{6}{7}$ &
$t+1$ \\
\midrule

$L_S$ & Event
& 
&
& Larvae hatch
& Larvae die
& \cellcolor{gray!25}\begin{tabular}{c}
Larvae feed\\and get infected\\
(Census)
\end{tabular}
& $L_S\to N_S$
& \\

  & Equation
& $0$
& $0$
& $\frac{r\frac{A\bigg(t-\frac{4}{7}\bigg)}{2}}{\frac{\frac{A\bigg(t-\frac{4}{7}\bigg)}{2}}{H_6\bigg(t-\frac{4}{7}\bigg)}+K}$
& $s_L S_L\left(t+\frac{3}{7}\right)L_S\left(t+\frac{3}{7}\right)$
& \cellcolor{gray!25}\begin{tabular}{c}
$L_S\left(t+\frac{4}{7}\right) \cdot$\\
$\bigg(1-\sum_{i=1}^5p_{i,L}\bigg(t+\frac{4}{7}\bigg)\cdot$\\$\beta_{i,L}\frac{H_{i_I}\left(t+\frac{4}{7}\right)}{H_i\left(t+\frac{4}{7}\right)}\bigg)$
\end{tabular}
& $0$
& $0$ \\

\midrule

$L_I$ & Event
& 
&
& 
& Larvae die
& \cellcolor{gray!25}\begin{tabular}{c}
Larvae feed\\and get infected\\
(Census)
\end{tabular}
& $L_I\to N_I$
& \\

  & Equation
& $0$
& $0$
& $0$
& $0$
& \cellcolor{gray!25}\begin{tabular}{c}
$L_I\left(t+\frac{4}{7}\right)+$\\
$L_S\left(t+\frac{4}{7}\right) \cdot$\\
$\bigg(\sum_{i=1}^5p_{i,L}\bigg(t+\frac{4}{7}\bigg)\cdot$\\$\beta_{i,L}\frac{H_{i_I}\left(t+\frac{4}{7}\right)}{H_i\left(t+\frac{4}{7}\right)}\bigg)$
\end{tabular}
& $0$
& $0$ \\

\midrule

$N_S$ & Event
& Nymphs die
& \cellcolor{gray!25}\begin{tabular}{c}
Nymphs feed\\and get infected\\
(Census)
\end{tabular}
&
$N_S\to A_S$

&
&
& 
$L_S\to N_S$

& \\

  & Equation
& $s_NS_N(t)N_S(t)$
& \cellcolor{gray!25}\begin{tabular}{c}$N_S\left(t+\frac{1}{7}\right) \cdot$\\
$\bigg(1-\sum_{i=1}^5 p_{i,N}\bigg(t+\frac{1}{7}\bigg)\cdot$\\$\beta_{i,N}\frac{H_{i_I}\left(t+\frac{1}{7}\right)}{H_i\left(t+\frac{1}{7}\right)}\bigg)$
\end{tabular}
& $0$
& $0$
& $0$
& $L_S\left(t+\frac{5}{7}\right)$
& $N_S\left(t+\frac{6}{7}\right)$\\

\midrule

$N_I$ & Event
& Nymphs die
& \cellcolor{gray!25}\begin{tabular}{c}
Nymphs feed\\and get infected\\
(Census)
\end{tabular}
&
$N_I \to A_I$

&
&
& 
$L_I \to N_I$

& \\

  & Equation
& $s_NS_N(t)N_I(t)$
& \cellcolor{gray!25}\begin{tabular}{c}
$N_I\left(t+\frac{1}{7}\right)+$\\
$N_S\left(t+\frac{1}{7}\right)\cdot$\\
$\bigg(\sum_{i=1}^5 p_{i,N}\bigg(t+\frac{1}{7}\bigg)\cdot$\\$\beta_{i,N}\frac{H_{i_I}\left(t+\frac{1}{7}\right)}{H_i\left(t+\frac{1}{7}\right)}\bigg)$
\end{tabular}
& $0$
& $0$
& $0$
& $L_I\left(t+\frac{5}{7}\right)$
& $N_I\left(t+\frac{6}{7}\right)$\\

\midrule

$A_S$ & Event
&
&
& $N_S\to A_S$
&
& \cellcolor{gray!25}\begin{tabular}{c}
Adults die\\
(Census)
\end{tabular}
& \begin{tabular}{c}Female adults\\feed and lay eggs\end{tabular}

& \\

  & Equation
& $0$
& $0$
& $N_S\left(t+\frac{2}{7}\right)$
& $A_S\left(t+\frac{3}{7}\right)$
& \cellcolor{gray!25}$s_A S_A\left(t+\frac{4}{7}\right)A_S\left(t+\frac{4}{7}\right)$
& $A_S\left(t+\frac{5}{7}\right)$
& $0$
\\

\midrule

$A_I$ & Event
&
&
& $N_I\to A_I$
&
& \cellcolor{gray!25}\begin{tabular}{c}
Adults die\\
(Census)
\end{tabular}
& \begin{tabular}{c}Female adults feed\\and lay eggs\end{tabular}
& \\

  & Equation
& $0$
& $0$
& $N_I\left(t+\frac{2}{7}\right)$
& $A_I\left(t+\frac{3}{7}\right)$
& \cellcolor{gray!25}$s_A S_A\left(t+\frac{4}{7}\right)A_I\left(t+\frac{4}{7}\right)$
& $A_I\left(t+\frac{5}{7}\right)$
& $0$
\\

\bottomrule
\end{tabular}}
\caption{Visual representation of how the tick–host dynamics are divided into seven sequential events. Gray cells denote census times where the state variables are defined.}
\label{timeline}
\end{sidewaystable}

\newpage
\begin{sidewaystable}
\setlength{\tabcolsep}{3pt}
\renewcommand{\arraystretch}{1.1}
{\begin{tabular}{l|l|ccccccc}
\toprule
 & &
$t+\frac{1}{7}$ &
$t+\frac{2}{7}$ &
$t+\frac{3}{7}$ &
$t+\frac{4}{7}$ &
$t+\frac{5}{7}$ &
$t+\frac{6}{7}$ &
$t+1$ \\
\midrule
$H_{i_S}$ & Event
& 
& \cellcolor{gray!25}\begin{tabular}{c}
Nymphs feed\\and infect hosts\\
(Census)\\
\end{tabular}
& Hosts born
& Hosts die
& 
& Hosts born
& Hosts die\\

  & Equation
& $H_{i_S}(t)$
& \cellcolor{gray!25}$H_{i_S}\left(t+\frac{1}{7}\right)e^{-\frac{p_{i,N}\left(t+\frac{1}{7}\right)\alpha_iN_I\left(t+\frac{1}{7}\right)}{H_{i}\left(t+\frac{1}{7}\right)}}$
& $H_{i_S}\left(t+\frac{2}{7}\right)+\frac{\Lambda_i}{2}$
& $H_{i_S}\left(t+\frac{3}{7}\right)e^{-\frac{\mu_i}{2}}$
& $H_{i_S}\left(t+\frac{4}{7}\right)$
& $H_{i_S}\left(t+\frac{5}{7}\right)+\frac{\Lambda_i}{2}$
& $H_{i_S}\left(t+\frac{6}{7}\right)e^{-\frac{\mu_i}{2}}$
\\

\midrule

$H_{i_I}$ & Event
& 
& \cellcolor{gray!25}\begin{tabular}{c}
Nymphs feed\\and infect hosts\\
(Census)
\end{tabular}
& 
& Hosts die
& 
& 
& Hosts die\\

  & Equation
& $H_{i_I}(t)$
& \cellcolor{gray!25}\begin{tabular}{c}
$H_{i_I}\left(t+\frac{1}{7}\right)+$\\
$H_{i_S}\left(t+\frac{1}{7}\right)\left(1-e^{-\frac{p_{i,N}\left(t+\frac{1}{7}\right)\alpha_iN_I\left(t+\frac{1}{7}\right)}{H_{i}\left(t+\frac{1}{7}\right)}}\right)$
\end{tabular}
& $H_{i_I}\left(t+\frac{2}{7}\right)$
& $H_{i_I}\left(t+\frac{3}{7}\right)e^{-\frac{\mu_i}{2}}$
& $H_{i_I}\left(t+\frac{4}{7}\right)$
& $H_{i_I}\left(t+\frac{5}{7}\right)$
& $H_{i_I}\left(t+\frac{6}{7}\right)e^{-\frac{\mu_i}{2}}$
\\

\midrule

$H_{5_S}$ & Event
& 
& \cellcolor{gray!25}\begin{tabular}{c}
Nymphs feed\\and infect hosts\\
(Census)
\end{tabular}
& Hosts born
& Hosts die
& 
&
& Hosts die\\

  & Equation
& $H_{5_S}(t)$
& \cellcolor{gray!25}$H_{5_S}\left(t+\frac{1}{7}\right)e^{-\frac{p_{5,N}\left(t+\frac{1}{7}\right)\alpha_5N_I\left(t+\frac{1}{7}\right)}{H_{5}\left(t+\frac{1}{7}\right)}}$
& $H_{5_S}\left(t+\frac{2}{7}\right)+\Lambda_5$
& $H_{5_S}\left(t+\frac{3}{7}\right)e^{-\frac{\mu_5}{2}}$
& $H_{5_S}\left(t+\frac{4}{7}\right)$
& $H_{5_S}\left(t+\frac{5}{7}\right)$
& $H_{5_S}\left(t+\frac{6}{7}\right)e^{-\frac{\mu_5}{2}}$
\\

\midrule

$H_{5_I}$ & Event
& 
& \cellcolor{gray!25}\begin{tabular}{c}
Nymphs feed\\and infect hosts\\
(Census)
\end{tabular}
&
& Hosts die
& 
&
& Hosts die\\

  & Equation
& $H_{5_I}(t)$
& \cellcolor{gray!25}\begin{tabular}{c} $H_{5_I}\left(t+\frac{1}{7}\right)+$\\
$H_{5_S}\left(t+\frac{1}{7}\right)\left(1-e^{-\frac{p_{5,N}\left(t+\frac{1}{7}\right)\alpha_5N_I\left(t+\frac{1}{7}\right)}{H_{5}\left(t+\frac{1}{7}\right)}}\right)$
\end{tabular}
& $H_{5_I}\left(t+\frac{2}{7}\right)$
& $H_{5_I}\left(t+\frac{3}{7}\right)e^{-\frac{\mu_5}{2}}$
& $H_{5_I}\left(t+\frac{4}{7}\right)$
& $H_{5_I}\left(t+\frac{5}{7}\right)$
& $H_{5_I}\left(t+\frac{6}{7}\right)e^{-\frac{\mu_5}{2}}$
\\

\midrule

$H_6$ & Event
&
&
& Hosts born
& Hosts die
& \cellcolor{gray!25}(Census)
&
& Hosts die\\

  & Equation
& $H_6(t)$
& $H_6\left(t+\frac{1}{7}\right)$
& $H_6\left(t+\frac{2}{7}\right)+\Lambda_6$
& $H_6\left(t+\frac{3}{7}\right)e^{-\frac{\mu_6}{2}}$
& \cellcolor{gray!25}$H_6\left(t+\frac{4}{7}\right)$
& $H_6\left(t+\frac{5}{7}\right)$
& $H_6\left(t+\frac{6}{7}\right)e^{-\frac{\mu_6}{2}}$
\\
\bottomrule
\end{tabular}}
\caption*{Table~\ref{timeline} continued.}
\end{sidewaystable}

\section{Basic Reproductive Number Derivation}\label{ap:BRN}
To derive the basic reproductive number when populations approach a stable fixed point (denoted here by $\mathcal{R}_{0_1}$), we linearize the infection subsystem of Model \eqref{fullsystem} around the disease-free fixed point. The infection subsystem is given by
\begin{equation}
\begin{aligned}
    L_I\left(t+1\right)&=\frac{r\frac{A_{\infty}}{2}}{\frac{\frac{A_{\infty}}{2}}{H_{6_{\infty}}}+K}s_LS_{L_{\infty}}\left(\sum_{i=1}^5 p_{{i,L}_{\infty}} \beta_{i,L} \frac{H_{i_I}\left(t\right)}{H_{i_{\infty}}+\frac{\Lambda_i}{k_i}}\right)\\
    N_I\left(t+1\right)&=s_NS_{N_\infty}\left(L_I\left(t\right)+L_S\left(t\right)\left(\sum_{i=1}^5p_{{i,N}_\infty}\beta_{i,N}\frac{H_{i_I}\left(t\right)e^{-\mu_i}}{H_{i_\infty}}\right)\right)\\
    A_I\left(t+1\right)&=s_AS_{A_\infty}s_NS_{N_\infty}\left(L_I\left(t\right)+L_S\left(t\right)\left(\sum_{i=1}^5p_{{i,N}_\infty}\beta_{i,N}\frac{H_{i_I}\left(t\right)e^{-\mu_i}}{H_{i_\infty}}\right)\right)\\
    H_{i_I}\left(t+1\right)&=\left(H_{i_I}\left(t\right)+ \left(H_{i_S}(t)+\lambda_i\right)\left(1-e^{-\frac{p_{{i,N}_\infty}\alpha_{i}s_NS_{N_\infty}L_I\left(t\right)}{H_{i_\infty}}}\right)\right)e^{-\mu_i},
\end{aligned}
\end{equation}
where
\begin{align*}
    \lambda_i = \frac{\Lambda_i}{2}\left(1+e^{\frac{\mu_i}{2}}\right) \quad \text{for $i \in \{1,2,3,4\}$}, \quad
    \lambda_5 = \Lambda_5.
\end{align*}
At the disease-free fixed point, all hosts are susceptible and tick survival and contact parameters take their equilibrium values $S_{j_\infty}$ and $p_{{i,j}_{\infty}}$. 

Because the infected larvae and host equations do not depend on $N_I$ or $A_I$, these equations decouple, and we can analyze the larvae-host subsystem separately. 

The Jacobian matrix of this system evaluated at the disease-free fixed point is
\begin{align*}
    J = \begin{bmatrix}
    \vspace{0.2cm}
        0&R_{L1}&R_{L2}&R_{L3}&R_{L4}&R_{L5}\\
        \vspace{0.2cm}
        R_{1L}e^{-\mu_1}&e^{-\mu_1}&0&0&0&0\\
        \vspace{0.2cm}
        R_{2L}e^{-\mu_2}&0&e^{-\mu_2}&0&0&0\\
        \vspace{0.2cm}
        R_{3L}e^{-\mu_3}&0&0&e^{-\mu_3}&0&0\\
        \vspace{0.2cm}
        R_{4L}e^{-\mu_4}&0&0&0&e^{-\mu_4}&0\\
        \vspace{0.2cm}
        R_{5L}e^{-\mu_5}&0&0&0&0&e^{-\mu_5}
    \end{bmatrix}
\end{align*}
where
\begin{align*}
    R_{Li}&=\frac{r\frac{A_{\infty}}{2}}{\frac{\frac{A_{\infty}}{2}}{H_{6_{\infty}}}+K}s_LS_{L_\infty} \frac{p_{{i,L}_\infty}\beta_{i,L}}{H_{i_{\infty}}+\frac{\Lambda_i}{k_i}},\\
    R_{iL}&=p_{{i,N}_\infty}\alpha_is_NS_{N_\infty}\frac{\left(H_{i_\infty}+\lambda_i\right)}{H_{i_\infty}}.
\end{align*}
The quantity $R_{Li}R_{iL}$ represents one complete larva-host-larva transmission loop involving host group $i$. Summing across host groups therefore accounts for the expected number of infected larvae produced by a larva that survives through the tick life cycle.

Following the next-generation operator method for discrete-time systems \citep{allen2008basic}, we split the Jacobian into $F$ and $T$ matrices. Here, $F$ contains the terms associated with new infections. The resulting matrix is
\begin{align*}F=
    \begin{bmatrix}
    \vspace{0.2cm}
        0&R_{L1}&R_{L2}&R_{L3}&R_{L4}&R_{L5}\\
        \vspace{0.2cm}
        R_{1L}e^{-\mu_1}&0&0&0&0&0\\
        \vspace{0.2cm}
        R_{2L}e^{-\mu_2}&0&0&0&0&0\\
        \vspace{0.2cm}
        R_{3L}e^{-\mu_3}&0&0&0&0&0\\
        \vspace{0.2cm}
        R_{4L}e^{-\mu_4}&0&0&0&0&0\\
        \vspace{0.2cm}
        R_{5L}e^{-\mu_5}&0&0&0&0&0
    \end{bmatrix},
\end{align*}
where
\begin{align*}
    R_{Li}&=\frac{r\frac{A_{\infty}}{2}}{\frac{\frac{A_{\infty}}{2}}{H_{6_{\infty}}}+K}s_LS_{L_\infty} \frac{p_{{i,L}_\infty}\beta_{i,L}}{H_{i_{\infty}}+\frac{\Lambda_i}{k_i}},\\
    R_{iL}&=p_{{i,N}_\infty}\alpha_is_NS_{N_\infty}\frac{H_{i_\infty}+\lambda_i}{H_{i_\infty}}.
\end{align*}
The quantity $R_{Li}R_{iL}$ represents one complete larva-host-larva transmission loop involving host group $i$. Summing across host groups therefore accounts for the expected number of infected larvae produced by a larva that survives through the tick life cycle.

The transition matrix $T$ captures host turnover and is given by
\begin{align*}
    T=\begin{bmatrix}
        0&0&0&0&0&0\\
        0&e^{-\mu_1}&0&0&0&0\\
        0&0&e^{-\mu_2}&0&0&0\\
        0&0&0&e^{-\mu_3}&0&0\\
        0&0&0&0&e^{-\mu_4}&0\\
        0&0&0&0&0&e^{-\mu_5}
    \end{bmatrix}.
\end{align*}
The next generation matrix is therefore
\begin{align*}
    Q=F(I-T)^{-1},
\end{align*}
where $I$ is the $6\times 6$ identity matrix.

The basic reproductive number $\mathcal{R}_0$ is the spectral radius of $Q$:
\begin{align*}
    \mathcal{R}_0=\rho(Q).
\end{align*}

Solving for the dominant eigenvalue yields
\begin{align*}
    \mathcal{R}_0 &=  \sqrt{
        \sum_{i=1}^5 R_{Li}R_{iL}\frac{e^{-\mu_i}}{1-e^{-\mu_i}}}.
\end{align*}
Solving the inequality $\mathcal{R}_0<1$ yields
\begin{align*}
    \sum_{i=1}^5R_{iL}R_{Li}\frac{e^{-\mu_i}}{1-e^{-\mu_i}}<1,
\end{align*}
which can be written as 
\begin{align*}
    \sum_{i=1}^5R_iS_{L_\infty}S_{N_\infty}\frac{e^{-\mu_i}}{1-e^{-\mu_i}}<1,
\end{align*}
with $\displaystyle R_i=\frac{R_{iL}R_{Li}}{S_{L_\infty}S_{N_\infty}}$.

The magnitude of $\mathcal{R}_0$ therefore depends on two interacting mechanisms: the efficiency of the larva-host-larva transmission cycle and the probability that hosts survive long enough to participate in repeated transmission events.

The disease-free fixed point is LAS whenever this inequality is satisfied. In this case, the combined contribution of larva-host-larva transmission loops is insufficient to overcome losses in the tick life cycle.

\section{Existence of Endemic Fixed Point}\label{ap:ee}
We now state a necessary and sufficient condition for the existence of a unique endemic fixed point of Model~\eqref{fullsystem}.
\begin{theorem}
Model~\eqref{fullsystem} has a unique endemic fixed point if and only if the demographic survival fixed point is locally asymptotically stable (i.e., $\mathcal{R}_d>1$) and $\mathcal{R}_0>1$.
\end{theorem}
\begin{proof}
We first reduce Model~\eqref{fullsystem} to a single equation for the infected larval population. Assuming all other populations are at equilibrium, we express $H_{i_I}$ as functions of $L_I$ and substitute into the $L_I$ equation to obtain
\begin{align*}
    L_I=f(L_{I}),
\end{align*}
where 
\begin{align*}
    f(L_I)=\frac{r\frac{A_{\infty}}{2}}{\frac{\frac{A_{\infty}}{2}}{H_{6_{\infty}}}+K}s_LS_{L_\infty}\left(\sum_{i=1}^5 p_{{i,L}_{\infty}} \beta_{i,L} \frac{H_{i_{\infty}}+\lambda_i}{H_{i_\infty}+\frac{\Lambda_i}{k_i}}\left(\frac{\left(1-e^{-\frac{p_{{i,N}_\infty}\alpha_{i}s_NS_{N_\infty}L_{I}}{H_{i_\infty}}}\right)e^{-\mu_i}}{1-e^{-\frac{p_{{i,N}_\infty}\alpha_{i}s_NS_{N_\infty}L_{I}}{H_{i_\infty}}}e^{-\mu_i}}\right)\right).
\end{align*}
The function $f: \mathbb{R}_{\ge0}\to\mathbb{R}$ is continuous and satisfies $f(0)=0$. 

Differentiation shows that $f'(L_I)>0$ and $f''(L_I)<0$ for all $L_I \geq 0$, implying that $f$ is strictly increasing and concave down.

Evaluating the derivative at the origin yields 
\begin{align*}
f'(0)=\frac{r\frac{A_{\infty}}{2}}{\frac{\frac{A_{\infty}}{2}}{H_{6_{\infty}}}+K}s_LS_{L_\infty}s_NS_{N_\infty}\left(\sum_{i=1}^5 p_{{i,L}_{\infty}}p_{{i,N}_{\infty}}\alpha_i \beta_{i,L}\frac{H_{i_\infty}+\lambda_i}{H_{i_\infty}+\frac{\Lambda_i}{k_i}}\frac{e^{-\mu_i}}{1-e^{-\mu_i}}\right)=\mathcal{R}_0^2
\end{align*}
Moreover, direct computation shows that 
\begin{align*}
    \lim_{L_I \to \infty} f'(L_I)=0.
\end{align*}
Thus, there exists $x>0$ such that $f'(L_I)<\frac{1}{2}$ for all $L_I \geq x$.
By the Mean Value Theorem, for $L_I\geq x$ we have
\begin{align*}
    f(L_I) \leq f(x)+\frac{1}{2}(L_I-x).
\end{align*}
Thus, for sufficiently large $L_I$,  
\begin{align*}
    f(L_I)<L_I,
\end{align*}
so the graph of $f$ eventually lies below the line $f(L_I)=L_I$.

Consequently, since $f$ is increasing, concave down, and satisfies $f(0)=0$, a unique positive endemic fixed point exists if and only if $f'(0)>1$, i.e.,
\begin{align*}
    \mathcal{R}_0^2>1.
\end{align*}
This analysis also tells us that the endemic fixed point is LAS when it exists since concavity implies $0<f'(L_{I_\infty})<1$ where $L_{I_\infty}>0$.
\end{proof}

\section{Basic Reproductive Number Under Density-Independent Host-Finding Success}\label{ap:BRN_constandsuccess}
To assess the impact of ratio-dependent host-finding success in our model, we compare the resulting $\mathcal{R}_0$ value with a model that assumes density-independent host-finding success (Figure~\ref{R0_rSK_constant}). Here, we take $s_LS_L(t)=0.1$, $s_NS_N(t)=0.1$, and $s_AS_A(t)=0.3$ for all $t$ \citep{daniels2000estimating, ostfeld2023mouse}. This assumption leads to the same expression for $\mathcal{R}_0$ discussed in Section~\ref{BRN_FP} and Appendix~\ref{ap:BRN}, with $S_{L_\infty}$ and $S_{N_\infty}$ replaced with their constant values. 
\begin{figure*}[ht!]
\centering
\includegraphics[width=\linewidth]{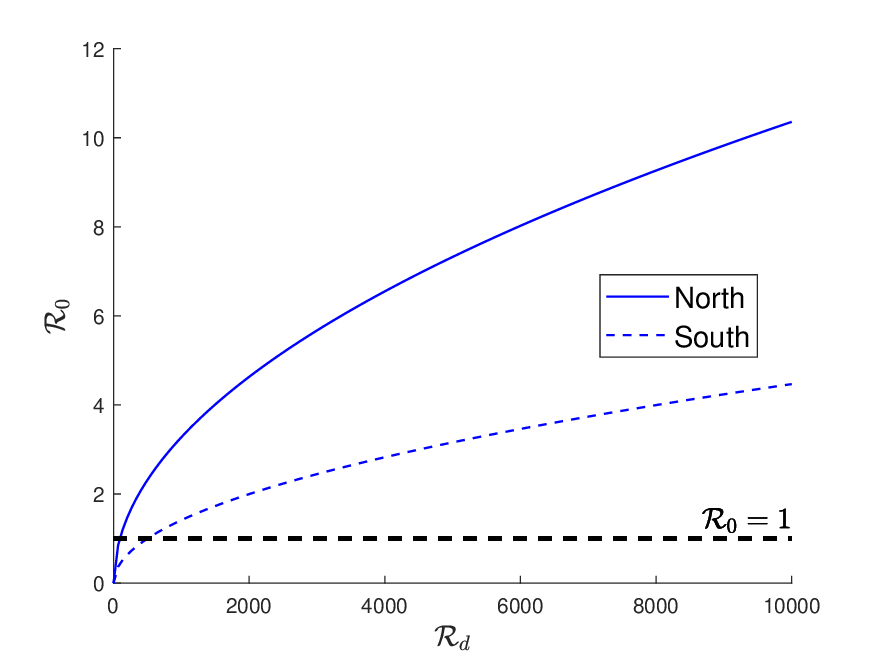}
    \caption{Plots of $\mathcal{R}_0$ as a function of the demographic reproductive number $\mathcal{R}_d=\frac{rS}{K}$ under the northern (solid) and southern (dashed) parameter sets from Table~\ref{paramvaluesnorthsouth}, respectively, under the assumption of density-independent host-finding success. In the figure, blue lines indicate $\mathcal{R}_0$ values evaluated at the coexistence fixed point.}
    \label{R0_rSK_constant}
\end{figure*}

\end{appendices}


\bibliography{sn-bibliography}

\end{document}